\documentclass[a4paper,UKenglish,cleveref, autoref, thm-restate]{lipics-v2021}

\pdfoutput=1 
\hideLIPIcs  

\title{On (Directed) Width-Parameters of Geometric Spanners} 

\author{Kevin Buchin}{TU Dortmund, Germany}{kevin.buchin@tu-dortmund.de}{https://orcid.org/0000-0002-3022-7877}{}
\author{Carolin Rehs}{ TU Dortmund, Germany \and TU Eindhoven, The Netherlands\and \url{https://carolin.rehs.me/}}{carolin.rehs@tu-dortmund.de}{https://orcid.org/0000-0002-8788-1028}{Partly supported by the PRIME program{me} of the German Academic Exchange Service (DA{AD}) with funds from the German Federal Ministry of Research, Technology and Space (BM{F}TR)}
\author{Torben Scheele}{TU Dortmund, Germany}{torben.scheele@tu-dortmund.de}{https://orcid.org/0009-0006-6119-6598}{Funded by the Deutsche Forschungsgemeinschaft (DFG, German Research Foundation)  -- project number 550144388.}

\authorrunning{K.~Buchin, C.~Rehs, T.~Scheele} 

\Copyright{Kevin Buchin, Carolin Rehs, Torben Scheele} 



\category{} 

\relatedversion{} 

\acknowledgements{
We want to thank all participants of KWCG 2025 for the helpful discussions on bounded dilation path spanners and Sariel Har-Peled for pointing out \cite{Matousek1990} to us.}

\nolinenumbers 

\EventEditors{John Q. Open and Joan R. Access}
\EventNoEds{2}
\EventLongTitle{42nd Conference on Very Important Topics (CVIT 2016)}
\EventShortTitle{CVIT 2016}
\EventAcronym{CVIT}
\EventYear{2016}
\EventDate{December 24--27, 2016}
\EventLocation{Little Whinging, United Kingdom}
\EventLogo{}
\SeriesVolume{42}
\ArticleNo{23}

\usepackage{hyperref}
\usepackage{amsthm} 
\usepackage{cleveref}
\usepackage{thmtools} 
\usepackage{thm-restate}

\usepackage{tcolorbox}

\usepackage[noend]{algpseudocode}
\usepackage{algorithm}

\algrenewcommand\algorithmicrequire{\textbf{Input:}}
\algrenewcommand\algorithmicensure{\textbf{Output:}}
\algnewcommand\algorithmicforeach{\textbf{for each}}
\algdef{S}[FOR]{ForEach}[1]{\algorithmicforeach\ #1\ \algorithmicdo}

\theoremstyle{plain}

\newcommand{\lab} {\text{lab}}
\newcommand{\CW}{\mbox{CW}}

\crefname{lemmastar}{Lemma}{Lemmas}
\Crefname{lemmastar}{Lemma}{Lemmas}

\crefname{theoremstar}{Theorem}{Theorems}
\Crefname{theoremstar}{Theorem}{Theorems}

\crefname{observationstar}{Observation}{Observations}
\Crefname{observationstar}{Observation}{Observations}

\crefname{propositionstar}{Proposition}{Propositions}
\Crefname{propositionstar}{Proposition}{Propositions}

\crefname{corollarystar}{Corollary}{Corollaries}
\Crefname{corollarystar}{Corollary}{Corollaries}

\crefname{fact}{Fact}{Lemmas}
\Crefname{fact}{Fact}{Lemmas}

\declaretheorem[numbered=no]{statement}{}

\newcommand{\dtw}{\text{dtw}}
\newcommand{\dpw}{\text{dpw}}
\newcommand{\dw}{\text{dw}}
\newcommand{\tw}{\text{tw}}
\newcommand{\pw}{\text{pw}}
\newcommand{\bw}{\text{bw}}
\newcommand{\cw}{\text{cw}}
\newcommand{\rw}{\text{rw}}
\newcommand{\clw}{\text{clw}}

\newcommand{\td}{\text{td}}

\newcommand{\bigO}{\mathcal{O}}
\renewcommand\P{{\sf P}}

\usepackage{wrapfig}
\usepackage{float}

\begin{document}

\maketitle

\begin{abstract}
To speed up algorithms on geometric graphs, it is common to approximate the complete Euclidean graph while maintaining certain geometric properties. A (directed) $t$-spanner $G$ for a point set $P$ in the Euclidean space is a (directed) graph such that for every pair of points, the shortest path in $G$ is at most a factor $t$ longer than the Euclidean distance between those points.

In this paper, we investigate $t$-spanners that are bounded by certain graph parameters. Let $\kappa$ be a graph parameter. We show that for path-width, branch-width and cut-width there is an $\mathcal{O}(n/k^{d/(d-1)})$-spanner $G$ on $P$ with $\kappa(G)=k$ and that this is asymptotically worst-case optimal.
In $\mathbb{R}^2$ we show the same bounds for planar graphs of clique-width or rank-width $k$. 
In contrast, for tree-depth, we show that there are sets of points for which the dilation cannot be bounded. Therefore, we investigate computing a spanner with tree-depth $k$ and minimum dilation. We show that already for tree-depth $3$ this problem is NP-hard to approximate within any factor strictly less than $\sqrt{2}$, and present an XP-algorithm to compute for a given tree-depth $k$ a graph with dilation at most $2t^*$, where $t^*$ is the minimum dilation. 

We further extend these results to obtain directed $\mathcal{O}(n/k^{d/(d-1)})$-spanners $G$ with $\kappa(G)=k$ for $\kappa$ being directed tree-width, directed path-width or DAG-width and show that also in the directed case, this is asymptotically worst-case optimal.

\end{abstract}

\section{Introduction}

Geometric graphs are of high interest for many applications, such as road and train networks, robot motion planning, mobile ad-hoc networks and many more. To work with and run algorithms on these graphs, it is often useful to approximate the complete Euclidean graph with a sparser subgraph, while maintaining certain geometric properties. In particular, the shortest-path distance between two points should not get much larger than their Euclidean distance: graphs with that property are called geometric spanners. 

There has been research on geometric spanners for multiple decades, see \cite{DBLP:journals/comgeo/BoseS13, DBLP:books/daglib/0017763} for surveys. In particular planar spanners, spanners with bounded edge number or spanners with bounded degree have been investigated. 
Most recently, at SoCG 2025, geometric spanners with bounded tree-width have been investigated \cite{DBLP:conf/compgeom/BuchinRS25}. From an algorithmic point of view, spanners with bounded tree-width admit polynomial time algorithms for many in general NP-hard problems~\cite{DBLP:journals/jct/RobertsonS83, DBLP:journals/jal/RobertsonS86}. However, in \cite{DBLP:conf/compgeom/BuchinRS25}, while giving a $\bigO(n/k^{d/(d-1)})$-spanner of tree-width $k$, the authors also show that there are sets of points on which there is a lower bound of $\Omega(n/k^{d/(d-1)})$ for the dilation of a tree-width $k$ spanner. 
Therefore, it seems most natural to investigate further graph parameters, hoping to obtain different bounds. 

In this paper, we study branch-width, path-width, cut-width, clique-width, rank-width, and tree-depth. Interestingly, we show that the tight trade-off between tree-width and dilation exists for several other parameters. Most surprisingly, we obtain exactly the same upper and lower bounds for more general and more restrictive parameters such as path-width, cut-width and branch-width. 
Since the clique-width for the complete graph is $2$ and the rank-width is $1$, bounding these parameters does not yield an approximation of a geometric graph. However, in $\mathbb{R}^2$, we investigate these parameters together with the common criterion of planarity and again obtain the same upper and lower bounds. 

In contrast to that, the parameter of tree-depth stands out: There are sets of points for which the dilation of a spanner with tree-depth $k$ cannot be bounded and therefore we cannot give an algorithm with dilation guarantees. We therefore investigate the problem of computing a spanner with tree-depth $k$ and minimum dilation. We show that this problem in general metric spaces is NP-hard to approximate within any factor strictly less than $\sqrt{2}$, and we present an XP-algorithm with respect to tree-depth that computes a graph with tree-depth $k$ and a dilation that 2-approximates the minimum dilation. 

For many applications, directed graphs can be even more useful to model geometric data. For instance, road networks include one-way roads, communication networks often rely on one-way communication lines. Hence, also directed and oriented spanners recently obtained some attention \cite{DBLP:conf/compgeom/BuchinKMORSW25, DBLP:journals/algorithmica/BuchinGKPRRW26}. In this paper, we therefore also investigate directed graph parameters. While it seems natural that asymptotically similar constructive results can be obtained by replacing every edge in an undirected spanner with a bidirected arc, one would expect that smaller constructions are possible. Instead, we can show that there are sets of points in $\mathbb{R}^d$, such that constructing a spanner $D$ of directed tree-width, directed path-width or DAG-width $k$ yields a lower bound of $\Omega(n/k^{d/(d-1)})$ for the dilation of $D$. 
An overview of the dilation bounds for (directed) spanners bounded in the different parameters can be found in the appendix.

\section{Technical Contribution}

From a technical point of view, it seems most surprising that the dilation lower bounds for tree-width can be extended to directed tree-width. While the result looks quite similar to the undirected cases, we need some structural insights that have not been known before. 
Specifically, the challenge is that directed tree-width as defined in \cite{JRST01} is not closed under contracting bidirected edges. We show that, using a slightly different definition for directed tree-width~\cite{DBLP:conf/soda/GiannopoulouKKK22} that is equivalent up to a constant factor to the definition in~\cite{JRST01}, directed tree-width is indeed closed under strongly connected contractions (\cref{lem:dtw_contraction}). By that, we obtain a lower bound for the dilation of directed tree-width (\cref{thm:dtw_lower}).

To obtain a spanner with bounded path-width, we show that there is a trade-off between path-width and dilation using similar techniques as in \cite{DBLP:conf/compgeom/BuchinRS25} or \cite{Aronov2008}. However, to use those techniques we need a path with a dilation that is bounded at most linearly in the number of input points. To obtain this, we use a result of \cite{Matousek1990} that embeds an $n$-point metric into the real number line, obtaining a slightly better dilation by a more careful analysis (\cref{mat3.2}). 

While the results for several other graph parameters follow from relations between those parameters, tree-width and path-width, the parameter of tree-depth stands out. We construct a point set where every geometric spanner with tree-depth $k$ has dilation at least $1+2/\varepsilon$ (for $\varepsilon > 0$). Hence, we cannot hope for an algorithm that gives dilation guarantees, but turn our attention to obtaining the best possible dilation. 
However, as we show by a reduction from the \textsc{SetCover} problem, for points in a general metric space even deciding whether there is a $2$-spanner of tree-depth 3 is NP-hard. Our construction implies that the problem is NP-hard to approximate within any factor strictly less than $\sqrt{2}$.
Thus, we look into constant-factor approximation, giving an XP-algorithm with respect to tree-depth that computes a spanner with tree-depth $k$ in $\bigO(2^{k}n^{2^{k}-1})$ time and dilation at most $2 t^*$, where $t^*$ is the minimum dilation (\cref{thm:tree-depth-xp}). To show this, we first give a polynomial-time algorithm that, for some $t \geq 1$, tests whether there is a spanner of tree-depth $k$ using Steiner points that has dilation at most $t$ (\cref{thm:tree-depth-xp_Steiner}).

\section{Preliminaries}

In the following we denote an undirected simple graph as $G=(V,E)$ and use $V(G)$ for the set of vertices of $G$ and $E(G)$ as the set of edges of $G$. By $\Delta(G)$ we denote the maximum degree of $G$. A directed graph or digraph is denoted as $D=(V,E)$. Given a set of points $\P\subset \mathbb{R}^d$ a geometric spanner $G$ with dilation $t$ is a graph on $\P$ s.t. for every pair of points $p,q\in \P$, it holds that $d_G(p,q)\leq t \cdot \|pq\|$, where $\|pq\|$ is the Euclidean distance between $p$ and $q$ and $d_G(p,q)$ is the sum of the distances along the edges of a shortest path from $p$ to $q$ in $G$. A directed geometric spanner is defined analogously for digraphs $D$. 

The definitions, related work and short discussions on the used graph parameters can be found in the appendix. For surveys on (directed) graph parameters see \cite[Chapter 6]{DE14} and \cite{Bodlaender24}. Definitions of tree-width and path-width \cite{DBLP:journals/jct/RobertsonS83, DBLP:journals/jal/RobertsonS86}, branch-width \cite{RobertsonS91}, cut-width \cite{Chung85}, clique-width \cite{CO00}, rank-width \cite{Oum17}, tree-depth \cite{Nesetril}, directed tree-width \cite{DBLP:conf/soda/GiannopoulouKKK22}, directed path-width \cite{Bar06}, and DAG-width \cite{BDHK06,Obd06} can be found in the named literature.  
Some definitions and notations are needed to understand proofs - those are explicitly given before used.

\section{Bounded Path-Width Spanners}

Path-width as a parameter is strongly related to tree-width. It is known that computing a geometric spanner of bounded tree-width with minimum dilation is NP-hard, even for tree-width~1 \cite{Klein2007, CHEONG2008188}. The same is true for bounded path-width spanners, even for path-width~1 \cite{GKM07}. By definition, every path decomposition for some graph $G$ is also a tree decomposition, since every path is a tree. 
Therefore, the lower bound for tree-width presented in \cite{DBLP:conf/compgeom/BuchinRS25} extends to path-width. However, since path-width is significantly more restricted than tree-width (as even for trees the path-width is not bounded by a constant), it seems difficult to obtain the same upper bound. Surprisingly, it is possible to find a spanner with path-width bounded by $k$ and the same dilation as for the tree-width, as we show in the following:

\begin{theorem}
\label{thm:pwupper}
    Given a set of $n$ points $\sf P\subset\mathbb R^d$ and some positive integer $k\leq n^{1-1/d}$, there is a geometric spanner of path-width $k$ and dilation $\bigO(n/k^{d/(d-1)})$.
\end{theorem}

The algorithm is based on a trade-off between path-width and dilation and includes similar ideas as the algorithm in~\cite{DBLP:conf/compgeom/BuchinRS25}. However, it is not possible to use that approach directly. To obtain a bounded tree-width spanner, the authors in~\cite{DBLP:conf/compgeom/BuchinRS25} start with a Euclidean minimum spanning tree, which given $n$ points has dilation $\bigO(n)$ and tree-width $1$. But since the path-width of a minimum spanning tree is not bounded through a constant (it can possibly be $\Omega(\log n)$), for a similar approach, it is necessary to compute a path spanner with a dilation that is linearly bounded in the number of points.
In \cite{Matousek1990}, Matoušek showed that it is possible to embed an $n$-point metric into the real number line with distortion $12n$, which means that for every point set in Euclidean space, a path of dilation $12n$ can be computed in $\bigO(n^2)$ time. In \cref{mat3.2}, we restate the statement in terms of our definitions and, by a more careful analysis, obtain a dilation of $4n$.

Given a set of points $\sf P$ we use $\text{EMST}(\sf P)$ to denote the Euclidean minimum spanning tree of $\sf P$ and given a geometric graph $G$ we use $G[\leq s]$ to denote the geometric graph obtained from $G$ by deleting all edges of length $>s$. Going further, we assume that the distance between two points in $\P$ is at least 1.

\begin{restatable}{lemmastar}{matousek}
\label{mat3.2}

Let ${\sf P}\subset \mathbb{R}^d$ be a set of points. There is a path $\pi$ containing all points of $\P$ such that for every $s$ that is a power of two the induced subgraphs $\pi[{\sf P}_i]$ are connected where ${\sf P}_1,\dots,{\sf P}_m$ denote the points of the connected components of $(\text{EMST}({\sf P}))[\leq s]$. The length of every path $\pi[{\sf P}_i]$ is $\leq 2|\P_i|\cdot s$ and the dilation is~$\leq 4|{\sf P}_i|$.
\end{restatable}{\let\thefootnote\relax\footnotetext{\hspace{-3.1mm}$^\ast$ \hspace{0.4mm}The proof of this result can be found in the appendix.}}
To compute a spanner of bounded path-width the algorithm starts by computing the EMST and splits it into multiple subtrees. These trees are then replaced by paths. Lastly, a spanner connecting these paths is computed.

\begin{algorithm}[H] \caption{}
\label{alg:higherdim}
\begin{algorithmic}[1]

\If {$k=1$}
\State {\Return path of Lemma \ref{mat3.2} on $\P$ with dilation $\leq 4n$}
\EndIf
\State $T_\text{MST} \leftarrow $ EMST(${\P}$)
\State $m \leftarrow \left\lceil \left(k/C\right)^{d/(d-1)} +1\right\rceil$ where $C$ is some constant that is chosen later.
\State Compute a set $\mathcal{T}$ of $m$ disjoint subtrees of $T_\text{MST}$, each containing ${O}(n/m)$ points, as for example explained in the proof of Lemma 8 in \cite{DBLP:conf/compgeom/BuchinRS25}.
\State For $T \in \mathcal{T}$ let $R(T)$ be the vertices in $T$ incident to edges removed by the previous step, i.e. the endpoints of the edges of $T_\text{MST}$ that connect two distinct disjoint subtrees of $\mathcal{T}$.
\State For each $T \in \mathcal{T}$ iterate through its edges $e$ from long to short. If $e$ lies on a path between two vertices in $R(T)$, remove the edge and thereby split $T$ into two new subtrees $T'$ and $T''$. Remove $T$ from $\mathcal{T}$ and insert $T'$ and $T''$. Set $R(T')=R(T)\cap V(T')$ and $R(T'')=R(T)\cap V(T'')$. This is done until for every $T\in \mathcal{T}$ we have $\vert R(T)\vert=1$.
\State For each $T\in \mathcal{T}$ let $\pi(T)$ be the path of Lemma \ref{mat3.2} on $V(T)$
\State $E'\leftarrow \bigcup_{T\in\mathcal{T}}\pi(T)$
\State $(R, E'') \leftarrow$ the greedy $3/2$-spanner for $R = \bigcup_{T\in\mathcal{T}}R(T)$
\State \Return $G=({\P},E'\cup E'')$
\end{algorithmic}
\end{algorithm}

To analyse the path-width of the spanner constructed by Algorithm \ref{alg:higherdim}, we use the fact that the greedy spanner has sublinear separators \cite{DBLP:journals/talg/LeT24} and a recursive construction for a path-decomposition (\cite{DBLP:journals/tcs/Bodlaender98}, Theorem 20 (iii)).

\begin{lemma}
\label{lem:greedypw}
Let ${\sf P}\subset \mathbb{R}^d$ be a set of $n$ points and $G$ the greedy $t$-spanner of $\sf P$ for some $t\in(1,3/2]$. $G$ has path-width $\bigO(n^{1-1/d})$.
\end{lemma}
\begin{proof}
For $d=1$ the spanner $G$ must be a path and therefore has path-width $\bigO(1)$. Let $d\geq 2$. By the results of Le and Than \cite{DBLP:journals/talg/LeT24} we know that every subgraph of $G$ on $k$ vertices has a $(1-\frac{1}{\eta_d2^{d+1}})$-balanced separator of size $c_{d} (t-1)^{1-2d} k^{1-1/d}$, where $c_{d}\leq 2^{\bigO(d)}$ is a constant and $\eta_d$ is the packing constant of $d$-dimensional Euclidean space. Using the well known recursive construction for path-decompositions using separators by Bodlaender (see  \cite{DBLP:journals/tcs/Bodlaender98}, Theorem 20 (iii)), we obtain a path-decomposition for $G$ of width bounded by the following recursive function:
\[\ w(n)=c_{d} (t-1)^{1-2d} n^{1-1/d}+w\left(\left(1-\frac{1}{\eta_d2^{d+1}}\right)\cdot n\right)=\bigO(n^{1-1/d})\]
The recursive function is bounded by $\bigO(n^{1-1/d})$ as for constant $d$ there exist $c>0$ and $\epsilon>0$ s.t. $c_d(t-1)^{1-2d}n^{1-1/d}>c\cdot n^\epsilon$.
\end{proof}

 We call the paths computed in line 8 of the algorithm the \textit{subtree-paths} and the points of $R$ the \emph{representatives}. The set of representatives has size at most $2m - 2$, since each representative is incident to one of the $m - 1$ edges removed. Thus, the greedy spanner $(R, E'')$ in line 10 is a greedy spanner on at most $2m - 2$ points.
By Lemma \ref{lem:greedypw} and an appropriate choice of the constant $C$, the path-width of $(R, E'')$ is at most $k-2$. After line~7 of the algorithm each of the subtrees $T \in \mathcal{T}$ contains only one representative and therefore each subtree-path contains one representative. Thus, the graph $G$ consists of $(R,E'')$ where each vertex is contained in an additional path. Let $(X_1,...,X_t)$ be a path-decomposition of $(R, E'')$ of width $k-2$ and $\pi_1,...,\pi_{2m-2}$ be the subtree-paths. For each $\pi_i$ let $(V_{i,1},...,V_{i,t_i})$ be a path-decomposition of width $1$. A path-decomposition of $G$ can be constructed as follows: For each $\pi_i$ we locate a bag $X_{j_i}$ of $(X_1,...,X_t)$ that contains the representative of $\pi_i$ and construct the path-decomposition $(V_{i,1}\cup X_{j_i},\dots, V_{i,t_i}\cup X_{j_i})$. All the bags of this decomposition have size at most $k+1$. To construct a path-decomposition for $G$ we simply insert the bags $V_{i,1}\cup X_{j_i},\dots, V_{i,t_i}\cup X_{j_i}$ after $X_{j_i}$ in $(X_1,...,X_t)$ for all $\pi_i$. It is not hard to see that this is a valid path-decomposition of width $\leq k$ for $G$.

\begin{figure}[htbp]
  \centering
  \includegraphics[page=2]{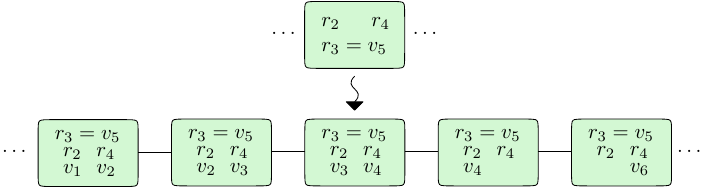}
  \caption{Construction of a path-decomposition for the spanner in Figure \ref{fig:1} based on the path-decomposition of the greedy spanner.}
  \label{fig:2}
\end{figure}

\begin{lemma}
    Given a set $\P\subset \mathbb{R}^d$ of $n$ points and some positive integer $k\leq n^{1-1/d}$, Algorithm~\ref{alg:higherdim} computes a spanner of dilation $\bigO(n/k^{d/(d-1)})$.
\end{lemma}
\begin{proof}
Let $p$ and $q$ be points in the same subtree-path $\pi$. Let $s$ be the length of the longest edge on the path between $p$ and $q$ in EMST($\P$) rounded up to the next power of two. It is clear that $p$ and $q$ are in the same connected component $S_{p,q,s}$ of $(\text{EMST}(\P))[\leq s]$. By Lemma~\ref{mat3.2} the path $\pi[V(S_{p,q,s})]$ must be connected and have length at most $2\cdot \bigO(n/k^{d/(d-1)})\cdot s$, as the number of points in the subtree-path is bounded by $\bigO(n/k^{d/(d-1)})$. It is also clear that $s\leq 2\|pq\|$, since otherwise $pq$ would have been shorter than an edge on the path between $p$ and $q$ in EMST($\P$) and this edge could have been replaced with $pq$ to obtain a tree of weight less than the weight of EMST($\P$).
So in total the distance between $p$ and $q$ in $G$ is 
\[2\cdot \bigO(n/k^{d/(d-1)})\cdot 2\|pq\|=\bigO(n/k^{d/(d-1)})\cdot \|pq\|\]
and the dilation therefore is $\bigO(n/k^{d/(d-1)})$.

Now let $p$ and $q$ be in different subtree-paths and let $\pi_p$ be the path containing $p$ and $\pi_q$ the path containing $q$. Let $s$ again be the length of the longest edge on the path between $p$ and $q$ in EMST($P$) rounded up to the next power of two. We claim that $p$ and the representative of $\pi_p$ are both contained in the same connected component $S_{p,s}$ of $(\text{EMST}(\P))[\leq s]$ and that $q$ and the representative of $\pi_q$ are both contained in the same connected component $S_{q,s}$ of $(\text{EMST}(\P))[\leq s]$. In other words: there is a path between $p$, respectively $q$, and its representative using only edges of length at most $s$. This is true since either the representatives of the subtree-paths are contained in the path between $p$ and $q$ in EMST($\P$) or if not then on this path there is an edge longer then all the edges on the path in EMST($\P$) between $p$ and the representative of $\pi_p$ or $q$ and the representative of $\pi_q$ respectively. This means that the edges from $p$ to its representative in EMST($\P$) consists of edges all shorter than $s$ and therefore all these edges must be in $(\text{EMST}(\P))[S_{p,s}]$. So by Lemma \ref{mat3.2} the subgraph $\pi_p[S_{p,s}]$ must be a connected path and contain both $p$ and the representative of $\pi_p$ and $\pi_q[S_{q,s}]$ must be a connected path containing both $q$ and the representative of $\pi_q$. Note that this argument uses the fact that $\text{EMST}(\P')$ is a subgraph of $\text{EMST}(\P)$ if $\text{EMST}(\P)[\P']$ is connected for $\P'\subset \P$. Furthermore, by Lemma~\ref{mat3.2} the length of these paths is at most $2\cdot \bigO(n/k^{d/(d-1)})\cdot s$. The distance of the path between the representatives of $\pi_p$ and $\pi_q$ is bounded by $\frac{3}{2}\cdot(\|pq\|+4\cdot \bigO(n/k^{d/(d-1)})\cdot s)$. Once again it is clear that $s\leq 2\|pq\|$. So in total the distance between $p$ and $q$ in $G$ is 
\[4\cdot \bigO(n/k^{d/(d-1)})\cdot 2\|pq\|+\frac{3}{2}\cdot(\|pq\|+4\cdot \bigO(n/k^{d/(d-1)})\cdot 2\|pq\|)=\bigO(n/k^{d/(d-1)})\cdot \|pq\|\]
and the dilation therefore is $\bigO(n/k^{d/(d-1)})$.
\end{proof}

\begin{figure}[htbp]
  \centering
  \includegraphics[page=1]{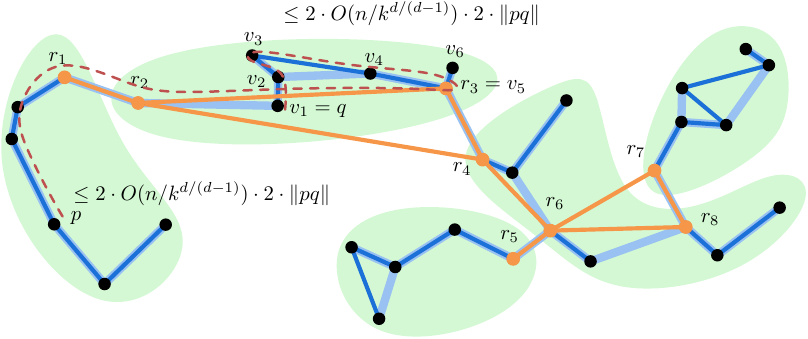}
  \caption{The subtree-paths (dark blue) on the original subtrees (green) and the greedy spanner (orange) as well as a shortest path (red) between points $p$ and $q$. The EMST is shown in light blue.}
  \label{fig:1}
\end{figure}
By relations between tree-width, path-width, cut-width and branch-width we then obtain the following result: 

\begin{corollary} \label{cutBranchwidth}
Given a set of $n$ points $\P\subset \mathbb{R}^d$ and some positive integer $k \leq n^{1-1/d}$, there is a geometric spanner of cut-width (resp. branch-width) $k$, and dilation $\bigO(n/k^{d/(d-1)})$. The dilation bound is asymptotically worst-case optimal. 
\end{corollary}
\begin{proof}
The dilation bound for spanners of bounded cut-width follows from the path-width bound: The cut-width of a graph $G$ with path-width $k$ is bounded by $\Delta(G)\cdot k$ \cite{DBLP:journals/dm/ChungS89} and the degree of the constructed spanner is bounded by the degree of the greedy spanner, which is constant \cite{DBLP:books/daglib/0017763} plus 2 for the decompositions of width 1. As the tree-width is a lower bound for the cut-width \cite{DBLP:journals/eatcs/Bodlaender88} the lower bound follows from the lower bound for tree-width bounded spanners \cite{DBLP:conf/compgeom/BuchinRS25}.

The dilation bound for spanners of bounded branch-width follows from the fact that there is a linear dependence between the tree-width $\tw(G)$ and branch-width $\bw(G)$, namely $\bw(G)-1\leq \tw(G)\leq \left\lfloor \frac{3}{2}\bw(G)\right\rfloor-1$ for $\bw(G)\geq 2$ \cite{RobertsonS91}.
In Section \ref{sec:td} we show that for spanners of branch-width 1 the dilation is unbounded, as the connected graphs of branch-width 1 are stars \cite{RobertsonS91}.
\end{proof}
Constructing a geometric $1$-spanner of low clique-width or rank-width is trivial, as every complete graph has clique-width 2 and rank-width $1$. With planarity additionally we have:
\begin{corollary}
\label{cliqueRankwidth}
    Given a set $\P\subset \mathbb{R}^2$ of $n$ points and some integer $k\in [4, 72 \sqrt{n-3}]$, a plane spanner with clique-width (resp. rank-width) $k$, maximum vertex degree $4$ and dilation $\bigO(n/k^2)$ can be constructed in $\bigO(n \log n)$ time. The dilation bound is asymptotically worst-case optimal for plane spanners. 
\end{corollary}
\begin{proof}
The result for bounded clique-width follows from the algorithm of Corollary 4 in~\cite{DBLP:conf/compgeom/BuchinRS25} for plane spanners of tree-width $k$ and the fact that a planar graph of tree-width $k$ has clique-width at most $6k-2$ \cite{DBLP:journals/dam/Courcelle20a}. The lower bound follows from the lower bound on the tree-width in \cite{DBLP:conf/compgeom/BuchinRS25} and the lower bound on the clique-width of $(k+1)/6$ for planar graphs of tree-width $k$ \cite{DBLP:conf/wg/GurskiW00}.

For rank-width, the upper bound on the dilation follows from the fact that for every graph $G$ it holds that the rank-width is bounded by the clique-width, i.e. $\rw(G)\leq \clw(G)$~\cite{OumS06}. The lower bound holds thanks to a result of Fomin, Oum and Thilikos \cite{FominOT10} who showed that for planar graphs $G$ it holds that $\tw(G)<72\cdot\rw(G)-1$.
\end{proof}

\section{Tree-Depth of Spanners}
\label{sec:td}

Contrary to the previously mentioned graph parameters, it is not possible to give an upper bound on the dilation of tree-depth bounded spanners:

\begin{theorem}
\label{thm-treedepth}
    Let $\varepsilon\in(0,1)$. For every positive integer $k$ there is a set of points $\P\subset \mathbb{R}$ s.t. any geometric spanner of tree-depth $k$ has dilation at least $1+2/\varepsilon$.
\end{theorem}
This follows from the following lemma and a lower bound on the tree-depth of paths: 
\begin{lemma}
    \label{lem:td_lemma}
    Let $\varepsilon\in(0,1)$, and $k\geq 1$ be some integer. There is a set of points $\P_k\subset \mathbb{R}$ s.t. any geometric spanner on $\P_k$ that does not contain a path of $2k$ vertices has dilation $\geq 1+2/\varepsilon$.
\end{lemma}
\begin{proof}
    We prove the statement by induction. For $k=1$, let $P_1$ be a set of two points of distance $1/\varepsilon$. As there cannot be an edge between the two points, the dilation must be larger than $1+2/\varepsilon$. For $k\geq 2$ we take two copies of $\P_{k-1}$ and place them next to each other $1/\varepsilon^k$ apart. Let $A$ be the first copy of $\P_{k-1}$ and $B$ the second, and let $G$ be a spanner on $\P_k$ of dilation less than $1+2\varepsilon$. We claim that the subgraphs of $G$ induced by $A$, respectively $B$, must be connected. If not, there must be two points in $A$, respectively $B$, next to each other that are connected only through a detour through the other copy. Let $d$ be the distance between these two points. The detour must have length at least $2/\varepsilon^{k}+d$ and thus the dilation is at least $1+2d/\varepsilon^k$. Since $d$ is at most $1/\varepsilon^{k-1}$, the dilation must be at least $1+2/\varepsilon$. Therefore, $G[A]$ and $G[B]$ are connected and by the induction hypothesis we know that the respective subgraphs of $G$ on $A$ and $B$ must contain paths of $2k-2$ vertices. Now let $pq$ be an edge that connects $A$ to $B$ and let $p\in A$ and $q\in B$. Inside the subgraph of $G$ induced by $A$ there must be a (possibly empty) path to some point of the path of length $2k-2$. This means that $p$ must be an endpoint of a path that contains at least $(2k-2)/2+1=k$ points of $A$ ($+1$ as $p$ itself is part of this path). The same holds for $q$ in $B$ analogously. As there is an edge between $p$ and $q$, the spanner $G$ must contain a path of at least $2k$ vertices.
\end{proof}

\begin{proof}[Proof of \Cref{thm-treedepth}]
    The tree-depth of an $n$ vertex path is $\lceil \log(n+1)\rceil$ and tree-depth is closed under taking minors \cite{Nesetril}. We choose $\P$ to be $\P_{2^{k-1}}$ from Lemma \ref{lem:td_lemma}. Let $G$ be a geometric spanner on $\P$ of dilation less than $1+2/\varepsilon$. We know that $G$ contains a path of $2^k$ vertices. Therefore, its tree-depth must be at least $\lceil \log(2^k+1)\rceil=k+1$.
\end{proof}

\begin{figure}[htbp]
  \centering
  \includegraphics[page=1]{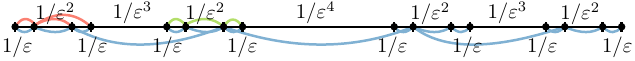}
  \caption{Inductive construction for lower bound for dilation of tree-depth bounded spanner. Red: spanner on $\P_2$ of tree-depth 1 but with dilation $2/\varepsilon+1$, green: spanner on $\P_2$ of tree-depth 2 and dilation 1, blue: spanner on $\P_4$ of tree-depth 3 and dilation $<2/\varepsilon+1$.}
  \label{fig:4}
\end{figure}
Therefore, we turn our attention to computing spanners with bounded tree-depth and optimum dilation. While already computing a minimum dilation tree is NP-hard~\cite{Klein2007,CHEONG2008188}, finding a spanner that is a star, i.e. a spanner of tree-depth $2$ can be solved naively in $\bigO(n^3)$ time by simply testing all vertices as centers and also in $\bigO(n2^{\alpha(n)}\log^2(n))$ expected time \cite{EppsteinW07}. We show that this does not extend to tree-depth $k$, as already finding a minimum dilation spanner with tree-depth $3$ is NP-hard. 

First recall the recursive definition of the tree-depth of a graph $G$ with connected components $G_1,\dots,G_\ell$, which we need in the following proofs. 
    \[\td(G)=\begin{cases}
        1& \text{if } |V(G)|=1\\
        1 + \min_{v\in V}\td(G-v) & \text{if $G$ is connected and $|V(G)|>1$}\\
        \max_{i\in\{1,\dots,\ell\}}\td(G_i) & \text{otherwise}
    \end{cases}\]

We further need the notion of \textit{level} for a vertex. Given a graph $G=(V,E)$ of tree-depth $k$, consider the following tree corresponding to the recursive definition of tree-depth. Each node contains a vertex of $G$ and can be associated with a subgraph of $G$. If the vertex contained in a node is removed from the associated subgraph, the tree-depth of the resulting connected components must be smaller by at least one. The node has one child for every connected component induced by the removal of the vertex and each child is associated with its connected component. The root of the tree is associated with $G$ itself and contains a vertex $v$ of $G$ that, if deleted, splits the graph into connected components of tree-depth at most $k-1$. The leaves of this tree are nodes associated with subgraphs consisting of only one vertex. We say that a vertex $v$ of $G$ has \textit{level} $i$ according to the recursive definition of tree-depth if the node containing $v$ is on level $i$ in the tree defined above. Notice that the tree above is not necessarily unique. Therefore, we fix one such tree for every graph.

Notice that in a graph there cannot be an edge between two vertices of the same level. The levels of the vertices also serve as the colors of a centered coloring, see \cite{Nesetril}.

\begin{theorem}
\label{thm:np-hard}
    Given a metric $(M,w)$, a tree-depth bound $k\geq 3$ and a dilation bound $t$, it is NP-hard to test whether there is a $t$-spanner of tree-depth $k$ on $(M,w)$. Furthermore, it is NP-hard to approximate the minimum dilation of a tree-depth $k$ spanner within any factor strictly less than $\sqrt{2}$.    
\end{theorem}

The above theorem is shown by proving NP-hardness for the following problem:

\begin{center}
\begin{tcolorbox}[width=10cm]
\begin{center}
 \begin{tabular}{rl}
\textbf{Problem:} & \textsc{Metric-TreeDepth-3-Spanner}\\
 \textbf {Given:} & A metric space $(M,w)$ and a dilation bound $t$.\\
 \textbf {Question:} & Is there a $t$-spanner of tree-depth 3 on $(M,w)$?\end{tabular}
\end{center}
\end{tcolorbox}
\end{center}

The NP-hardness is shown by a reduction from \textsc{SetCover}.

\begin{center}
\begin{tcolorbox}[width=12cm]
\begin{center}
 \begin{tabular}{rl}
\textbf{Problem:} & \textsc{SetCover}\\
 \textbf {Given:} & A universe $U=\{s_1,\dots,s_n\}$, a collection $\mathcal{S}=\{S_1,\dots,S_m\}$\\
 &of subsets of $U$ and a positive integer $k$.\\
 \textbf {Question:} & Is there a subset $\mathcal{C}\subseteq\mathcal{S}$ of $k$ sets so that $\bigcup_{C\in\mathcal{C}}C=U$? \end{tabular}
\end{center}
\end{tcolorbox}
\end{center}

\paragraph*{The Construction}
The metric $(M,w)$ is defined using a graph $G$ that consists of gadgets for the sets of the input set $\mathcal{S}$ and for the elements of the universe $U$. The metric closure of $G$ is the input metric $(M,w)$ for \textsc{Metric-TreeDepth-3-Spanner}. The dilation bound for the \textsc{Metric-TreeDepth-3-Spanner} problem is $t=2$. The constants $\alpha$ and $\beta$ in the construction can be chosen as $\alpha=\sqrt{2}-1\approx0.414$ and $\beta=2-\sqrt{2}\approx0.586$. We define the gadgets of $G$ as follows:

\begin{figure}[htbp]
  \centering
  \includegraphics[page=2]{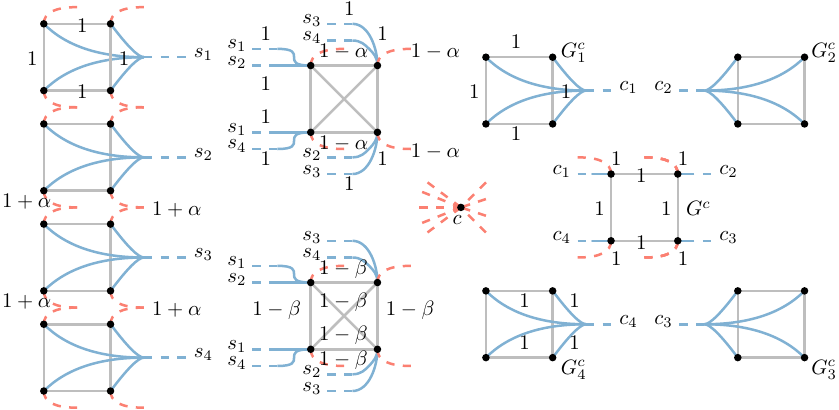}
  \caption{Illustration of the construction used for the reduction. The sets used in the example are $S_1=\{s_1,s_2\},S_2=\{s_3,s_4\},S_3=\{s_2,s_3\},S_4=\{s1_,s_4\}$.}
  \label{fig:np-hard2}
\end{figure}

\begin{enumerate}
    \item For every $s\in U$ a gadget of 4 points is constructed. These 4 points form a cycle where each edge has weight 1.

    \item For the sets $\mathcal{S}$ there are $k$ set gadgets in total. Each set gadget consists of a total of $\max\{m,4\}$ points which form a complete graph where every edge has weight $1-\beta$.
    \item The graph contains a special vertex $c$ that has an edge of weight $1-\alpha$ to the vertices of the set gadgets and an edge of weight $1+\alpha$ to every vertex of an element gadget. 

    \item There is also a centering gadget that makes sure that $c$ is the only vertex of level 1 in a spanner on $G$. This gadget consists of subgraphs $G^c_i$ for $i\in\{1,2,3,4\}$, where each $G^c_i$ consists of a cycle of 4 vertices whose edges all have weight $1$. There is also a subgraph $G^c$ that consists of a cycle of 4 vertices $\{c_1,c_2,c_3,c_4\}$ whose edges all have weight $1$. Each vertex $c_i$ has edges of weight $1$ to the vertices of $G^c_i$ as well as an edge to $c$.
\end{enumerate}

\paragraph*{Soundness and Correctness of the Reduction}

Observe that all gadgets together form a connected graph. As $G$ is connected, every spanner on $(M,w)$ must be connected and can have only one vertex of level 1. 

Before we can prove the soundness of the reduction we need a few auxiliary lemmas. The first lemma proves that the centering gadget works as intended.

\begin{lemma}
\label{lem:centered}
    Let $G'$ be a spanner of tree-depth 3 on $(M,w)$. The vertex $c$ must be the only vertex of level 1 in $G'$. If this does not hold the dilation of $G'$ must be at least $2\sqrt{2}$.
\end{lemma}
\begin{proof}
First we argue that the dilation for the vertices of the subgraphs $G^c_i$ for $i\in \{1,2,3,4\}$ can only be at most 2 if every vertex of $G^c_i$ has an edge to $c_i$. If one vertex of $G^c_i$ is a vertex of level 1 in $G'$ and one of level 2 in $G'$ then it is clear that the dilation in $G'$ for the vertices of $G^c_i$ is at most 2. But since not every $G^c_i$ can contain a vertex of level 1 the only way for the vertices of $G^c_i$ to have dilation at most 2 is if $c_i$ is a vertex of level 2 and has an edge to every vertex in $G_i^c$. A vertex of $G_i^c$ being level 2 in $G'$ does not suffice as in this case two vertices would have dilation 3. This means that at least three vertices of $G^c$ must be of level 2 in $G'$. We claim that in fact all 4 vertices must be of level 2. As at least three vertices are of level 2 and cannot have edges between them, the only way for a dilation of at most 2 between these vertices is via $c$. Otherwise, the dilation would be at least 3.  This means that all 4 vertices of $G^c$ must be of level 2 and $c$ must the single level 1 vertex.
\end{proof}

\begin{lemma}
\label{lem:set_gadget}
    Let $G'$ be a spanner of tree-depth 3 on $(M,w)$. Exactly one vertex of every set gadget is of level 2 in $G'$. If this does not hold the dilation of $G'$ must be at least $2\sqrt{2}$.
\end{lemma}
\begin{proof}
    By \cref{lem:centered} it is clear that the vertices of a set gadget cannot be of level 1. Now suppose that there is no vertex of level 2 in the set gadget. This means that all vertices of the set gadget must be connected through a vertex of level 2 outside the gadget or simply via $c$. As the distance between two vertices in a set gadget is $1-\beta$ and the edges to $c$ are of length $1-\alpha$ the dilation would be at least $(2-2\alpha)/(1-\beta)\geq2\sqrt{2}>2$. If they are connected via a vertex of an element gadget the dilation would be at least $2/(1-\beta)\geq2\sqrt{2}>2$. Finally, if the vertices are connected through a vertex of a different set gadget then the detour would be at least $2-2\alpha$ again leading to a dilation of at least $(2-2\alpha)/(1-\beta)\geq2\sqrt{2}>2$.
    Should there be multiple vertices of level two in a set gadget then there must be two vertices of the gadget that are connected via $c$ leading to a dilation of at least $(2-2\alpha)/(1-\beta)\geq2\sqrt{2}>2$.
\end{proof}

\begin{lemma}
\label{lem:element_gadget}
    Let $G'$ be a spanner of tree-depth 3 on $(M,w)$. Every vertex of an element gadget must be of level 3 in $G'$. If this does not hold the dilation of $G'$ must be at least $2\sqrt{2}$.
\end{lemma}
\begin{proof}
    By Lemma \ref{lem:centered} it is clear that the vertices of an element gadget cannot be of level 1. Now suppose that there are vertices of level 2. If there is one vertex of level 2 that has edges to the three other vertices of the gadget then there are two vertices of distance 1 that do not share an edge. The dilation for these two vertices would then be at least $3$. If there are two or more vertices of level 2 or there is a vertex of level 3 in the gadget that does not have an edge to a vertex inside the gadget then there must be vertices in the gadget of distance 1 whose detour must visit $c$ leading to a dilation of at least $2+2\alpha\geq2\sqrt{2}>2$. 
\end{proof}
    
With the above lemmas at hand we are ready to prove the correctness and soundness of the reduction.

\begin{lemma}
        If there is a 2-spanner of tree-depth 3 on $(M,w)$ then there must be a set cover $C$ of size $k$ for $\mathcal{S}$.
\end{lemma}
\begin{proof}
    We claim that the vertices of an element gadget must all have an edge to a vertex $v$ of a set gadget that has distance 1 to the vertices of the element gadget. By Lemma \ref{lem:element_gadget} and \ref{lem:set_gadget} it is clear that all vertices of an element gadget must be connected to a vertex of a set gadget which must be the same set gadget as otherwise the detour for two vertices of an element gadget of distance 2 must be via $c$ which would lead to a dilation of at least $1+1-\alpha+1+1-\alpha\geq2\sqrt{2}>2$. Now suppose the vertex $v$ does not have distance 1 to the vertices of the element gadget. Let $u$ and $u'$ be two vertices of distance 1 in the element gadget. The length of an edge from $v$ to $u$ or $u'$ is at least $1+1-\beta>1$. The dilation for $u$ and $u'$ therefore would be at least $2+2-2\beta\geq2\sqrt{2}>2$.

    There is a total of $k$ set gadgets each containing exactly one vertex of level 2 in $G'$. Every element gadget has edges to one of these level 2 vertices that has distance 1 to its vertices. By construction there are only edges of distance 1 between a vertex of a set gadget and an element gadget if the corresponding set contains the element of the element gadget. This means there are $k$ sets from $\mathcal{S}$ whose union contains every element of $U$. 
\end{proof}
From the above proof it is clear that if no set cover $C$ of size $k$ for $\mathcal{S}$ exists, then there either must be an element gadget whose vertices only have edges from a level 2 vertex that has distance 2 to the vertices of the gadget, leading to a dilation of $1+1-\alpha+1+1-\alpha\geq2\sqrt{2}$ or one of the properties from \cref{lem:centered}, \cref{lem:set_gadget} or \cref{lem:element_gadget} is violated, also leading to a dilation of $2\sqrt{2}$.

Lastly we prove correctness of the reduction and therefore Theorem \ref{thm:np-hard}.

\begin{lemma}
    If there is a set cover $C$ of size $k$ for $\mathcal{S}$ then there is a 2-spanner of tree-depth 3 on $(M,w)$.
\end{lemma}
\begin{proof}
    Let $S_{i_1},\dots, S_{i_k}$ be the sets of the cover $C$. We construct the $2$-spanner $G'$ as follows: We place an edge from $c$ to all vertices of the element and set gadgets as well as to the vertices of $G^c$. For the vertices of the centering gadget we place edges from $c_i$ to all vertices of $G_i^c$ for $i\in\{1,2,3,4\}$. Then for each $i\in\{i_1,\dots,i_k\}$ we pick a set gadget and place edges from the vertex corresponding to $S_{i}$ to all vertices in the same gadget as well as to all vertices of element gadgets for elements $s\in S_i$ that do not have edges yet.

    We now argue that the dilation of this spanner is at most 2 by a case distinction.
    \begin{itemize}
        \item Vertices in element gadgets: Let $v$ and $u$ be vertices in the same element gadget. $v$ and $u$ have either distance 1 or 2. Since there is a set that covers the element that corresponds to the gadget of $u$ and $v$ there are edges of length 1 connecting $u$ and $v$ to a vertex of a set gadget. The dilation for $u$ and $v$ is therefore at most 2. If $u$ and $v$ are in different element gadgets than their distance is at least 2. Both have an edge to the vertex $c$ of length $1+\alpha$. The dilation therefore is at most $1+\alpha\leq 2$.
        \item Vertices in set gadgets: Let $v$ and $u$ be vertices in the same set gadget. The distance between $v$ and $u$ is $1-\beta$. By the construction of $G'$ there either is a direct edge between $u$ and $v$ or there is a third vertex in the same set gadget that has an edge to both $u$ and $v$. As every edge in a set gadget has length $1-\beta$ the dilation between $u$ and $v$ is at most 2. If $u$ and $v$ are in different set gadgets then their distance is at least $2-2\alpha$. $u$ and $v$ are connected via edges of length $1-\alpha$ to $c$. The dilation therefore is at most $2$. 
        \item A vertex in a set gadget and a vertex in an element gadget: Let $u$ be a vertex of an element gadget and $v$ a vertex of a set gadget. The distance between $u$ and $v$ is at least 1. There is an edge from $u$ to $c$ of length $1+\alpha$ and an edge from $v$ to $c$ of length $1-\alpha$. The dilation therefore is at most $1+\alpha+1-\alpha=2$.
        \item Vertices in the centering gadget: Let $u$ and $v$ be vertices of $G^c$. The distance between $u$ and $v$ is at least $1$. Both vertices have an edge to $c$ of length one. The dilation therefore is at most $2$. Similarly if $u$ and $v$ are in some $G^c_i$ for $i\in\{1,2,3,4\}$ then their distance is at least 1 and both have edges of length 1 to $c_i$. Now let $v$ be a vertex of $G^c$ and $u$ a vertex of some $G^c_i$ for some $i\in\{1,2,3,4\}$. For symmetry reasons it suffices to only consider $v=c_1$. If $u$ is a vertex of $G^c_1$ then the distance between $u$ and $v$ is 1 and there is an edge of length 1 between $u$ and $v$ in $G'$. If $u$ is in $G^c_2$ or $G^c_4$ then the distance between $u$ and $v$ is 2 and there is a path via $c$ in $G'$ of length 3 which means the dilation is at most 2. If $u$ is in $G^c_3$ then the distance between $u$ and $v$ is 3 and there is a path via $c$ of length 3. 
        \item A vertex in the centering and a vertex outside the centering gadget: Let $u$ be a vertex of $G^c$. If $v$ is a vertex of an element gadget than the distance between $u$ and $v$ is at least $1+1+\alpha$ and there is a path via $c$ of length $1+1+\alpha$. If $v$ is a vertex of a set gadget than the distance is $1+1-\alpha$ and there is a path via $c$ of length $1+1-\alpha$. If $u$ is in some $G^c_i$ then the same bounds apply with an additional distance of 1 between the points and an additional edge of length 1 in the paths.
        \item The vertex $c$: As there is an edge from $c$ to all vertices of the element and set gadgets the dilation between $c$ and these vertices is 1. The same holds for vertices of $G^c$. For vertices in $G_i^c$ the distance is 2 and there is a path of length 2 via a vertex of $G^c$.
    \end{itemize}
Lastly we argue that the constructed spanner indeed is of tree-depth 3. It is not hard to see that if $c$ is deleted from $G'$ that each remaining connected component is a star. A star of course has tree-depth 2, which means that $G'$ must be of tree-depth $3$.\end{proof}

Since already for tree-depth $3$, it is NP-hard to compute a minimum dilation spanner, we cannot hope to give an XP- or even FPT-algorithm with respect to tree-depth for computing a minimum dilation spanner with bounded tree-depth. However, we present an XP-algorithm that computes a spanner with tree-depth $k$ and dilation at most $2 t^*$, where $t^*$ is the minimum dilation for a spanner of tree-depth $k$.

The idea is to first use an algorithm that decides whether there is a spanner of tree-depth $k$ and a certain dilation using Steiner points which recursively guesses points as star points for subsets of points and places an edge between every point in the subset and the star point. The star point is either a point of the given subset or a Steiner point in a face of a certain arrangement specified later. Given a star point, the algorithm computes subsets of points based on the dilation in the star. If the dilation for two points via the star point is not yet small enough, the two points get assigned to the same set and are connected through another star point in the next recursive call.

\begin{theorem} \label{thm:tree-depth-xp_Steiner}
    Given a set of $n$ points $\P\subset \mathbb{R}^2$, some positive integer $k$, and $t\geq 1$ there is an algorithm that tests whether there is a spanner of tree-depth $k$ using Steiner points that has dilation at most $t$ in time $\bigO(n^{4k+4}\log n)$.
\end{theorem}

\begin{theorem} \label{thm:tree-depth-xp}
    Given a set of $n$ points $P\subset \mathbb{R}^2$ and some positive integer $k$ there is an algorithm that computes a spanner of tree-depth $k$ and dilation at most two times the dilation of the minimum dilation spanner of tree-depth $k$ in time $\bigO(2^{k}n^{2^{k}-1})$.
\end{theorem}

Note that the dilation for spanners with Steiner points only considers pairs of points from the original set of points.

To compute the location of Steiner points, our algorithm makes use of a special arrangement of ellipses, which we define here. For two points $p,q\in \mathbb{R}^2$, we define the $t$-ellipse as the set $\{r\in \mathbb{R}^2\mid \|pr\|+\|rq\|= t\}$. For a set of points $\P\subset \mathbb{R}^2$, we define the arrangement $\mathcal{A}_t(\P)$ as the arrangement of all $t$-ellipses for pairs of points in $\P$.

\begin{algorithm}[h] \caption{}
\label{alg:td_steiner}
\begin{algorithmic}[1]
\Require a set of $n$ points ${\sf P}\subset \mathbb{R}^2$ and tree-depth bound $k\geq 2$ and a dilation bound $t$ 
\Ensure \textbf{true} if there is a spanner on $\P$ with Steiner points and dilation at most $t$ that is of tree-depth at most $k$, otherwise \textbf{false}
\State Calculate the Arrangement $\mathcal{A}_t(\P)$ and let $F$ be a set of points that contains a point from every face of $\mathcal{A}_t(\P)$
\State For each $c\in\P\cup F$, test if $\textsc{TDSpannerSteiner}(\P,c,t,{k-1})$. If \textbf{false} is returned for every call, \Return \textbf{false} else \Return \textbf{true}
\Procedure{$\textsc{TDSpannerSteiner}$}{$\P,c,t,m$}
  \State Create a graph $G$ with $\P-\{c\}$ as vertices and no edges.
  \State For every pair of points $p,q\neq c$ calculate the dilation in the star-graph on $\P$ centred at $c$. If the dilation is larger than $t$, place an edge between $p$ and $q$ in $G$.
  \State Let $\P_1,\dots,\P_\ell$ be the points of the connected components of $G$.
  \State \textbf{if} $m = 1$ \textbf{then} 

  \State \qquad \Return \textbf{ false} if there is a $\P_i$ of size $\geq 2$ else \textbf{true} 

  \State \textbf{for each} $\P_i$ \textbf{do} 
  \State \qquad Calculate the Arrangement $\mathcal{A}_t(\P_i)$ and let $F_i$ be a set of points that contains a point from every face of $\mathcal{A}_t(\P_i)$ 

  \State \Return $\bigwedge_{i = 1}^\ell\bigvee_{c'\in\P_i\cup F_i} \textsc{TDSpannerSteiner}(\P_i,c',t,m-1)$
  
\nolinenumbers \EndProcedure
\linenumbers
\end{algorithmic}
\end{algorithm}

The combinatorial complexity of the arrangement is clearly $\bigO(n^4)$, as there are $\bigO(n^2)$ ellipses and each ellipse intersects in at most 4 points with another ellipse. It can be computed in time $\bigO(n^4\log n)$, using the algorithm of Bentley and Ottmann \cite{BentleyO79}. The sets of points based on the dilation calculation can be computed in time $\bigO(n^2)$. Each call of \textsc{TDSpannerSteiner} triggers at most $n+\bigO(n^4)$ recursive calls. The recursion terminates at a depth of $k-1$. The total number of calls therefore is $\bigO(n^{4k})$ and the total running time therefore is $\bigO(n^{4k+4}\log n)$. \cref{thm:tree-depth-xp_Steiner} then follows from the following lemma:

\begin{lemma}
Given a set of points $\P\subset\mathbb{R}^2$, a tree-depth bound $k$ and a dilation bound $t$, the above algorithm tests whether there exists a spanner with Steiner points on $\P$ with dilation at most $t$ and of tree-depth at most $k$.
\end{lemma}
\begin{proof}
Let $G^*$ be a $t$-spanner of tree-depth $k$ with Steiner points and let $C_i^*\subset \P$ be the vertices in $\P$ of level $i$ in $G^*$. We can assume that $G^*$ is edge maximal. We define $G^*_i=(\P\cup \bigcup_{j=1}^iC^*_j,\{\{u,v\}\in E(G^*)\mid u\in \bigcup_{j=1}^iC^*_j\})$. This sequence of graphs in some sense is a bottom up construction of the graph $G^*$, using the recursive definition of tree-depth. Let $B_i^*$ be the set of pairs of points in $\P$ whose dilation in $G^*_i$ is at most $t$. We define $\mathcal{P}^*_i=\{\{p\in \P-\bigcup_{j=1}^i C^*_j\mid \{p,c\}\in E(G^*_i)\}\mid c\in C^*_i\}\cup \{\{c\in \P\}\mid c\in \bigcup_{j=1}^i C^*_j\}$ 
as a partition of $\P$ where two points of $\P$ are in the same part if they share a common vertex of level $i$ and are not in any $C^*_j$ for $j\leq i$. 

We claim that the algorithm presented above implicitly constructs a $t$-spanner $G$ of tree-depth $k$ with Steiner points. This can be shown by induction over the bottom-up construction of $G^*$ according to the recursive definition of tree-depth.
For each level $i$ of the induction we define $G_i$ as the current spanner constructed by the algorithm and $C_i$ as the vertices of $G_i$ chosen to be of level $i$ according to the definition of tree-depth. Based on these definitions we analogously define $B_i$ as the set of pairs of points in $\P$ whose dilation in $G_i$ is at most $t$, and further $\mathcal{P}_i=\{\{p\in \P-\bigcup_{j=1}^i C_j\mid \{p,c\}\in E(G_i)\}\mid c\in C_i\}\cup \{\{c\in \P\}\mid c\in \bigcup_{j=1}^i C_j\}$ as a partition of $\P$. We say that $\mathcal{P}_i$ is a refinement of $\mathcal{P}^*_i$, if for every part $\P'\in\mathcal{P}_i$ there is a part $\P^*\in\mathcal{P}^*_i$ so that $\P'\subseteq\P^*$. We now prove the following statement inductively:
\begin{statement}[H]
\begin{enumerate}
    \item For every $i$ it holds that $B_i^*\subseteq B_i$.
    \item For every $i$ it holds that $\mathcal{P}_i$ is a refinement of $\mathcal{P}^*_i$.
\end{enumerate}
\end{statement}
Observe that the statement above for $i=k-1$ implies that $G$ is a $t$-spanner.

For $i=1$ the graph $G^*_1$ is a star with centre $v^*$. This $v^*$ must be contained in some face of the arrangement $\mathcal{A}_t(\P)$. The algorithm will either choose a Steiner point in this face or if $v^*\in \P$ it will choose $v^*$ itself. Either way this means that $B^*_1=B_1$ which are exactly the pairs of points whose ellipses $v^*$ is contained in. It also holds that $\mathcal{P}_i=\mathcal{P}^*_i$.

Now suppose that $i>1$. First we are going to describe how $G_i$ is constructed. The algorithm above computes for every $\P_{i-1}\in\mathcal{P}_{i-1}$ maximal sets of points so that for every point $p$ in such a set there is a point $q$ in the same set for which the dilation of $p$ and $q$ is larger than $t$ in $G_{i-1}$. Let $\P_{i-1}^*\in \mathcal{P}^*_{i-1}$ and $v_1,\dots,v_\ell\in C^*_i$ be the vertices that split $\P_{i-1}^*$ into multiple smaller parts in $\mathcal{P}_{i}^*$. Let $\P_{i-1,1},\dots, \P_{i-1,h}\in\mathcal{P}_{i-1}$ be the sets that are subsets of $\P_{i-1}^*$, guaranteed by the fact that $\mathcal{P}_{i-1}$ is a refinement of $\mathcal{P}_{i-1}^*$. For every $\P_{i-1,j}$ let $Q_1,\dots,Q_r\subseteq\P_{i-1,j}$ be the maximal subsets found be the algorithm. For each of these sets we need to find a centre which we add to $C_i$. If a $v^*_{j}$ from $v_1,\dots,v_\ell$ is a original point from $\P$ and contained in some $Q_j$ the algorithm will choose it as a centre for $Q_j$. For all other $Q_j$ the algorithm will choose a Steiner point, which we specify later. The partition $\mathcal{P}_i$ can now be defined as the union of all $Q_j-v_j^*$ respectively $Q_j$, if the centre point is going to be a Steiner point, and that for all $\P_{i-1}\in\mathcal{P}_{i-1}$.

We are now going to argue that $\mathcal{P}_{i}$ is a refinement of $\mathcal{P}^*_{i}$. Suppose that there is a part $\P^*\in \mathcal{P}^*_i$ that is split into at least two parts that are not disjoint with some $\P^*\in\mathcal{P}_{i}$. As $\P^*$ is chosen by the algorithm to be a maximal set of points so that for every point $p\in\P^*$ there is a point $q\in \P^*$ for which the dilation of $p$ and $q$ is larger than $t$ in $G_{i-1}$ there must be parts $\P^*_1\in\mathcal{P}^*_{i}$ and $\P^*_2\in\mathcal{P}^*_{i}$ so that $p\in \P^*_1$ and $q\in \P^*_2$ where the dilation of $p$ and $q$ in $G_{i-1}$ must be larger than $t$. As $p$ and $q$ are in different parts of $\mathcal{P}^*_{i}$ their dilation in $G^*_{i-1}$ must have been at most $t$. As the dilation of points in different parts cannot decrease in $G^*_j$ for $j\geq i$. This is a contradiction to $B^*_{i-1}$ being a subset of $B_{i-1}$ and therefore there cannot be two parts in $\mathcal{P}^*_{i}$ that are not disjoint to a part in $\mathcal{P}_{i}$, proving that $\mathcal{P}_{i}$ is a refinement of $\mathcal{P}^*_{i}$.

Finally, we are going to complete the description of $G_i$ by adding Steiner points serving as new centres and argue that $B_i^*\subseteq B_i$. Let $v^*\in C^*_{i}$ be the vertex that is the centre of some $\P^*\in\mathcal{P}^*_{i}$. Suppose that $v^*$ is contained in $\P$. As mentioned earlier the algorithm will choose $v^*$ as well to serve as a centre for one of the sets $\P_{j_1},\dots, \P_{j_h}\in\mathcal{P}_i$ that are subsets of $\P^*$. To choose centres for the remaining sets the algorithm chooses Steiner-points in faces that are contained in a superset of the ellipses $v^*$ is contained in. This is possible as $\mathcal{A}_t(P^*)$ is a superarrangement of the arrangements of $\P_{j_1},\dots, \P_{j_h}$. The dilation for points in different sets of $\P_{j_1},\dots, \P_{j_h}$ must already be at most $t$. We claim that all the remaining pairs of points of $B^*_{i-1}\cap B^*_i$, whose dilation in $G_i^*$ was reduced to at most $t$ through the split by $v^*$, have dilation at most $t$ in $G_i$. This is true since every Steiner point chosen by the algorithm for $\P_{j_1},\dots, \P_{j_h}$ is contained in a superset of the ellipses $v^*$ is contained in. If $v^*$ is a Steiner point the same arguments hold, but the algorithm will instead of $v^*$ choose a Steiner point contained in a superset of ellipses that $v^*$ is contained in. Therefore, we have that $B^*_{i+1}\subseteq B_{i+1}$.
\end{proof}

To approximate the optimization problem of finding spanners of bounded tree-depth, we now use the algorithm for the decision problem of the variant with Steiner points as a subroutine and the fact that we can replace Steiner points with points of the original set while not increasing the dilation bound too much. To move from the decision problem to the optimization problem, the algorithm first computes a set of possible dilation bounds. As a graph of tree-depth $k$ cannot contain a path consisting of more than $2^k-1$ vertices \cite{Nesetril} this set can be computed by enumerating all paths of at most this length and the resulting dilation for the endpoints.

\begin{algorithm}[H] \caption{}
\label{alg:td}
\begin{algorithmic}[1]
\Require a set of $n$ points ${\sf P}\subset \mathbb{R}^2$ and a natural number $k\geq 2$
\Ensure at most two times the minimum dilation of a spanner $G$ of tree-depth $k$ on $\P$
\State For every pair of points calculate all paths of length at most $2^k-1$ and calculate the dilation bounds induced by every path. Let $B$ be the set of all these dilation bounds.
\State Sort $B$.
\State Perform binary search on $B$ for smallest $t$ for which Algorithm \ref{alg:td_steiner} returns \textbf{true}.
\State \Return $2t$

\end{algorithmic}
\end{algorithm}

\begin{lemma}
\label{lem:td_steiner_to_general}
    Let $G$ be a spanner on a set of points $\P\subset\mathbb{R}^2$ of tree-depth $k$ with Steiner points and dilation $t$ that was implicitly constructed by Algorithm \ref{alg:td_steiner}, then there is a spanner $G'$ on $\P$ of tree-depth $k$ and dilation at most $2t$ without Steiner points.
\end{lemma}
\begin{proof}
    Let $G$ be a spanner of tree-depth $k$ with Steiner points that was implicitly computed by Algorithm \ref{alg:td_steiner} and has dilation $t$. To construct the spanner $G'$ we simply contract every Steiner point to its closest neighbor in $G$ that is not a Steiner point. As $G$ is edge maximal, every Steiner point must have a neighbor that is not a Steiner point. We can assume that the shortest path between the points $p,q\in\P$ is via the point $c$ that is of the highest level among all points that have both $p$ and $q$ as neighbors. Suppose that $c$ is a Steiner point and $r$ is the closest neighbor of $c$. This means that $\|cr\|\leq \|cp\|$ and $\|cr\|\leq \|cq\|$. If $c$ is contracted to $r$, then the shortest path in $G'$ between $p$ and $q$ is of length at most $\|pc\|+2\|cr\|+\|cq\|\leq 2(\|pc\|+\|cq\|)$. The dilation of $G'$ is therefore at most $2t$. The tree-depth is closed under contractions \cite{Nesetril}. Therefore, the tree-depth of $G'$ must be at most~$k$ as well. 
\end{proof}

The dilation of the minimum dilation spanner of tree-depth $k$ with Steiner points is a lower bound for the minimum dilation of spanners of tree-depth $k$. Therefore, we know that the spanner of the lemma above has dilation at most two times the minimum dilation among all spanners of tree-depth $k$.

The total number of paths considered is bounded by $\bigO(n^{2^k-1})$ as the paths of length $2^{k}-1$ contain all paths of smaller length. Calculating the dilation bounds induced by a path takes $\bigO(2^k-1)$ time.
The running time used for the calculation of the set of possible dilation bounds therefore is $\bigO(2^{k}\cdot n^{2^{k}-1})$. In the binary search Algorithm \ref{alg:td_steiner} is called $\bigO(2^k\log n)$ times. The total running time therefore is $\bigO(2^{k}n^{2^{k}-1}+2^kn^{4k+4}\log^2 n)=\bigO(2^{k}n^{2^{k}-1})$. Together with \cref{lem:td_steiner_to_general}, this concludes the proof of \cref{thm:tree-depth-xp}.

\section{Spanners of Bounded Directed Width Parameters}

Historically, there have been many different definitions of tree-width for directed graphs, most of which are equivalent up to some constant factor. The notion of directed tree-width we are using is the one introduced by Giannopoulou et al. in \cite{DBLP:conf/soda/GiannopoulouKKK22}. Kim, Hatzel and Kreutzer~\cite{kim2025directedtreewidthclosedtaking} showed that this definition is equivalent up to a constant factor to all other common definitions of directed tree-width.

A digraph $D$ has directed tree-width at most $k$ if there is a triple $\mathcal{T}=(T, \mathcal{X}, \mathcal{W})$, where $T$ is a rooted directed tree, $\mathcal{X}: V(T)\rightarrow 2^{V(D)}$ and $\mathcal{W}: E(T)\rightarrow 2^{V(D)}$ such that: 
\begin{itemize}
     \item $\{\mathcal{X}(t) \mid t\in V(T)\}$ is a partition of $V(D)$ into possibly empty sets such that $\mathcal{X}(r)\neq \emptyset$, where $r$ is the root of $T$
     \item for all $e=(s,t)\in E(T)$, there is no closed walk in $D-\mathcal{W}(e)$ containing a vertex of $\mathcal{X}(T_t)$ and a vertex of $V(D)-\mathcal{X}(T_t)$, and
     \item $\max\{|\mathcal{X}(t)\cup \bigcup_{e\sim t} \mathcal{W}(e)|-1 \mid t\in V(T)\} \leq k$ where $e\sim t$ means that $t$ is an endpoint of the edge $e$.
\end{itemize}
Here, $T_t:=T[\{t'\in V(T)\mid \text{$t'$ is reachable from $t$ by a directed path in $T$}\}]$ for every $t\in V(T)$ and $\mathcal{X}(S):=\bigcup_{t\in V(S)}\mathcal{X}(t)$.
The triple $\mathcal{T}=(T, \mathcal{X}, \mathcal{W})$ is called a \emph{directed tree-decomposition} of $D$.

In this section we show that the lower bound on the dilation of tree-width bounded spanners in \cite{DBLP:conf/compgeom/BuchinRS25} can be extended to the directed setting.
\begin{theorem}
\label{thm:dtw_lower}
Let $d \geq 2$ be a fixed integer. For positive integers $n$ and $k \leq n^{(d-1)/d} \cdot(12d)^{1/d-2}$
there is a set of $n$ points in $\mathbb{R}^d$, such that every directed geometric spanner of directed tree-width $k$ on this set has dilation $\Omega(n/k^{d/(d-1)})$.
\end{theorem}
To prove this bound, we use a similar set of points as in \cite{DBLP:conf/compgeom/BuchinRS25} that will resemble the $(h+1)^d$-grid ($(h+1)\times ...\times(h+1)$-grid), where $h = \left\lceil (9d\cdot (k+2))^{1/(d-1)} -1\right\rceil$. We make use of the fact, that the bidirected version of this $d$-dimensional grid has large directed tree-width.
\begin{lemma}
\label{lem:twbidigrid}
    The bidirected $n^d$-grid has directed tree-width $\geq \frac{1}{9d}\cdot n^{d-1}-1$.
\end{lemma}
\begin{proof}
    Let $G$ be a $n^d$-grid. By Lemma 12 in \cite{DBLP:conf/compgeom/BuchinRS25} we know that $G$ must be of tree-width at least $\frac{2}{9d}\cdot n^{d-1}-1$. In \cite{DBLP:journals/dam/KozawaOY14} it was observed that every graph of tree-width $k$ must have a strict bramble of order at least $\frac{1}{2}(k+1)$. Therefore $G$ must have a strict bramble of order at least $\frac{1}{9d}\cdot n^{d-1}$. By Corollary 2.9 in \cite{kim2025directedtreewidthclosedtaking} every digraph with a strict bramble of order $k$ must have directed tree-width at least $k-1$. The strict bramble of $G$ directly corresponds to a strict bramble of the bidirected version of $G$. Therefore the bidirected $n^d$-grid must have a strict bramble of order $\frac{1}{9d}\cdot n^{d-1}$ and therefore must be of directed tree-width $\geq \frac{1}{9d}\cdot n^{d-1}-1$.
\end{proof}

The number of points that represent an edge is $m=\lfloor n/ (d\cdot (h+1)^d- d\cdot (h+1)^{d-1})\rfloor$. Formally, the set $\P_{d,n,k}$ is defined as $\P_{d,n,k}=\bigcup_{i=1}^d \P_i$, where
\[\P_{i}=\bigcup^{h}_{j_1=0}\dots \bigcup^{h}_{j_d=0}\{ (j_1\cdot m,\dots,j_{i-1}\cdot m,x_i,j_{i+1}\cdot m,\dots,j_d\cdot m) \mid x_i\in \{0, \dots, h\cdot m\}\}.\]

The proof in \cite{DBLP:conf/compgeom/BuchinRS25} for the lower bound on the dilation heavily relies on contractions of paths. Unfortunately, directed tree-width is not closed under arbitrary edge contraction. Instead we are using contractions of strongly connected subdigraphs. Using ideas from \cite{kim2025directedtreewidthclosedtaking} we can show that directed tree-width is closed under these contractions.
\begin{lemma}
\label{lem:dtw_contraction}
    Directed tree-width is closed under strongly connected contractions.
\end{lemma}
\begin{proof}
Let $D$ be a digraph of directed tree-width $k$ and $C$ a strongly connected subdigraph of $D$. We claim that the digraph $D'$ which is obtained by contracting $C$ to $x_C\notin V$ in $D$ has directed tree-width at most $k$. Let $c$ be some vertex in $V(C)$ and let $\mathcal{T}=(T,\mathcal{X},\mathcal{W})$ be a directed tree-decomposition of width $k$ for $D$. As $C$ is a subdigraph of $D$ there must be some $t_{c}\in V(T)$ for which $c\in \mathcal{X}(t_{c})$. We construct a directed tree-decomposition $\mathcal{T'} = (T, \mathcal{X}', \mathcal{W}')$ of width at most $k$ for $D'$ as follows:

    \begin{itemize}
        \item For all $e\in E(T)$, if for some $v\in V(C)$ we have $v\in \mathcal{W}(e)$, we let $\mathcal{W}'(e)=(\mathcal{W}(e)-V(C))\cup \{x_C\}$ and otherwise $\mathcal{W}'(e)=\mathcal{W}(e)$.
        \item For all $t\neq t_{c}$ in $V(T)$ let $\mathcal{X}'(t)=\mathcal{X}(t)-V(C)$ and for $t_{c}$ we set $\mathcal{X}'(t_{c})=(\mathcal{X}(t_{c})-V(C))\cup \{x_C\}$.
    \end{itemize}

It is not hard to see that $\{\mathcal{X}'(t) \mid t \in V(T)\}$ is a partition with possibly empty sets of $V(D)$. And since in the construction of $\mathcal{X}'$ and $\mathcal{W}'$ there were only vertices replaced and removed, the width of $\mathcal{T}'$ must be at most $k$.

It remains to show that the second condition for directed tree-width holds for $\mathcal{T}'$. For the sake of contradiction suppose that there is an edge $e=(s,t)\in E(T)$ for which there is a closed walk $W'$ in $D'-\mathcal{W}'(e)$ visiting a vertex of $\mathcal{X}'(T_t)$ and a vertex of $D-\mathcal{X}'(T_t)$. We distinguish two cases. Either $\mathcal{W}'(e)=\mathcal{W}(e)$ or $\mathcal{W}'(e)=(\mathcal{W}(e)-C)\cup \{x_C\}$. First consider the case where $\mathcal{W}'(e)=(\mathcal{W}(e)-C)\cup \{x_C\}$. Since $x_C\notin V(D')-\mathcal{W}'(e)$ and $C\cap V(D') = \emptyset$, the vertex $x_C$ as well as no vertex in $C$ is visited by $W'$. This means that $W'$ is also a closed walk in $D-\mathcal{W}(e)$ that visits a vertex in $\mathcal{X}(T_t)$ and a vertex in $D-\mathcal{X}(T_t)$, which is a contradiction. Now suppose that $\mathcal{W}'(e)=\mathcal{W}(e)$. The vertex $x_C$ may now be visited by the closed walk $W'$. First let us assume it is not visited by $W'$. As the closed walk $W'$ does not visit $x_C$ and $\mathcal{W}'(e)=\mathcal{W}(e)$, it must also be a closed walk in $D-\mathcal{W}(e)$ that visits a vertex in $\mathcal{X}(T_t)$ and a vertex in $D-\mathcal{X}(T_t)$, which is again a contradiction. Therefore, now let us suppose that $x_C$ is contained in the walk $W'$. There may be multiple occurrences of $x_C$ in this walk. We claim that by replacing each occurrence of $x_C$ in $W'$ by vertices of $C$ we obtain a walk $W$ in $D-\mathcal{W}(e)$ containing vertices from $\mathcal{X}(T_t)$ and $D-\mathcal{X}(T_t)$. Let us have a look at some occurrence of $x_C$ in $W'$. There must be some vertex $w$ before $x_C$ in $W'$ and some vertex $y$ after $x_C$ in $W'$. Since both edges $(w,x_C)$ and $(x_C,y)$ are in $D'-\mathcal{W}'(e)$, there must be edges $(w,c')$ and $(c'',y)$ in $D-\mathcal{W}(e)$, where $c',c''\in V(C)$. Since $C$ is strongly connected, there must be a directed path from $c'$ to $c''$ in $C$. We replace this occurrence of $x_C$ in $W'$ by a walk containing every vertex of $C$ starting at $c'$ and ending at $c''$. Once again such a walk exists since $C$ is strongly connected. Now this occurrence of $x_C$ is replaced by a walk that contains every vertex of $C$. If we replace every occurrence of $x_C$ in $W'$ in this way we obtain the walk $W$. By construction $W$ is a closed walk in $D-\mathcal{W}(e)$. We claim that this walk visits a vertex from $\mathcal{X}(T_t)$ and a vertex from $D-\mathcal{X}(T_t)$. Suppose that $x_C\in \mathcal{X}'(T_t)$. Then there must be some vertex $z\in D'-\mathcal{X}'(T_t)$ on $W'$. Since $x_C$ is in $\mathcal{X}'(T_t)$ we know that $c\in \mathcal{X}(T_t)$. And for $z$ we know that by construction $z\in D-\mathcal{X}(T_t)$. As $W$ visits both $c\in \mathcal{X}(T_t)$ and $z\in D-\mathcal{X}(T_t)$, we once again reached a contradiction.
\end{proof}
Using these contractions on a spanner on $\P_{d,n,k}$ we can show the following lemma:
\begin{lemma}
\label{lem:dtw_lower}
    If $D$ is a directed geometric $\bigO(n/k^{d/(d-1)})$-spanner on $\P_{d,n,k}$ then it must be of directed tree-width $>k$.
\end{lemma}
\begin{proof}
 Let $D$ be a directed geometric $\bigO(n/k^{d/(d-1)})$-spanner on $\P_{d,n,k}$. Let $\mathcal{R}_{x} := x+[-m/4, m/4]^d$ denote the hypercube of side-length $m/2$ centered at $x$. Let  $P_{\boxplus}=\{(j_1\cdot m, \dots , j_d\cdot m)\mid j_1,\dots, j_d \in \{0, \dots, h\}\}$ be the points of the underlying grid. For any pair $p,q\in P_{\boxplus}$ of neighboring grid points, we consider the three hypercubes $\mathcal{R}_{p}, \mathcal{R}_{q}$, and $\mathcal{R}_{s} := s+[-m/4,m/4]^d$, where $s= \frac{p+q}{2}$ is the midpoint between $p$ and $q$. 
    
Let $p_0 = p, p_1, p_2, \ldots p_m=q$ be the ordered sequence of points representing the edge of the grid between $p$ and $q$, and further let $s_i := \frac{p_{i-1} + p_{i}}{2}$.
We argue that for every $p$ and $q$ there is a sequence of closed walks $((W_i)_{i\in\{1,\dots,m\}})_{p,q}$ in $D$ where each $W_i$ is fully contained in $\mathcal{R}_{p_{i-1}}\cap \mathcal{R}_{p_{i}}$ and the union of all $W_i$ is strongly connected and contains $p$ and $q$.

 Since $D$ is a directed $\bigO(n/k^{d/(d-1)})$-spanner, we can assume that for every pair $p_{i-1}$, $p_{i}$ there is a shortest path $P_i$ from $p_{i-1}$ to $p_{i}$ in $D$ and a shortest path $Q_i$ from $p_{i}$ to $p_{i-1}$, both of length at most $m/4 = \Theta(n/k^{d/(d-1)})$. Both these paths necessarily will remain within $\mathcal{R}_{p_{i-1}}\cap \mathcal{R}_{p_{i}} \subset \mathcal{R}_{s_i}$.
For $1 \leq i \leq m/2$ we have $\mathcal{R}_{s_i} \subset  \mathcal{R}_{p} \cup \mathcal{R}_{s}$, and for $m/2 < i \leq m$ we have $\mathcal{R}_{s_i} \subset  \mathcal{R}_{q} \cup \mathcal{R}_{s}$.

For each $i\in \{1,\dots , m\}$ we define the closed walk $W_i$ as the concatenation of $P_i$ and $Q_i$. We have such a sequence $(W_i)_{i\in\{1,\dots,m\}}$ for any pair $p,q\in P_{\boxplus}$. We claim that this implies that $D$ can be transformed into a digraph isomorphic to the bidirected $(h+1)^d$-grid using only operations that do not increase its directed tree-width and must be therefore by Lemma~\ref{lem:twbidigrid} of directed tree-width $>k$. To show this, we perform a series of strongly connected contractions on $D$, obtaining a graph $D'$. 

Let $S_p:=\bigcup_{i\in \{1,\dots,\lfloor m/4 \}\rfloor} W_i$ be the union of the closed walks for points in $\mathcal{R}_p$. Notice that $S_p$ may contain points $p_i$ from $\mathcal{R}_s$ but only for $i< m/2$, as a walk $W_j$ may stretch out to all points in $\mathcal{R}_{s_j}$. We analogously define  $S_q:=\bigcup_{i\in \{\lceil{3m/4}\}\rceil, \dots, m\}} W_i$ and $S_{p,q}:=\bigcup_{i\in \{1, \dots, m\}} W_i$. Notice that $S_p$ and $S_q$ are disjoint and both strongly connected.  We perform a strongly connected contraction of $S_p$ to $p$ and $S_q$ to $q$. After performing the contractions the vertices of $V(S_{p,q})-(V(S_p)\cup V(S_q))$ form strongly connected components with $p$ or $q$, as $S_p$, $S_q$ and $S_{p,q}$ were strongly connected. We contract the strongly connected subgraphs induced by the vertices $V(S_{p,q})-(V(S_p)\cup V(S_q))$ to either $p$ or $q$. The resulting digraph contains edges $(p,q)$ and $(q,p)$. 
Further we observe that \emph{(i)} any vertex $p'$ of $S_{p,q}$ in $\mathcal{R}_{p}$ has been identified with $p$, \emph{(ii)} any vertex $q'$ of $S_{p,q}$ in $\mathcal{R}_{q}$ has been identified with $q$, and \emph{(iii)} any other vertex $r$ of $S_{p,q}$ is not contained in another subdigraph $S_{s,t}$ with $\{p, q\} \neq \{s,t\}$. The latter holds, since $r$ cannot lie in $\mathcal{R}_{s, t}
$. This means, that the subdigraphs $S_{p,q}$ and $S_{s,t}$ are contracted consistently.
Thus, regardless of the order in which we contract the subdigraphs, we obtain a digraph $D'$ which has a subdigraph isomorphic to the bidirected $(h+1)^d$-grid. In $D'$, there is a bidirected edge between every neighboring pair of grid points. Therefore the subdigraph obtained by removing  edges between non-neighboring grid points is isomorphic to the $(h+1)^d$-grid. Note that directed tree-width is closed under taking subdigraphs \cite{kim2025directedtreewidthclosedtaking}.
We conclude that the directed geometric spanner $D$ can be transformed to a digraph isomorphic to the $(h+1)^d$-grid by only using operations that do not increase the directed tree-width and by Lemma~\ref{lem:twbidigrid} therefore must be of directed tree-width at least $\frac{1}{9d}(h+1)^{d-1}-1\geq k+1$. 
\end{proof}

Based on this result and further insights for relations between directed tree-width, directed path-width and DAG-width we obtain the same bounds for the latter parameters: 

\begin{corollary}
 Given a set of $n$ points $\sf P\subset\mathbb R^d$ and some positive integer $k\leq n^{1-1/d}$, there is a directed geometric spanner of directed path-width (resp. DAG-width) $k$ and dilation $\bigO(n/k^{d/(d-1)})$. This dilation bound is asymptotically worst-case optimal. 
\end{corollary}
To prove this corollary, we need several technical insights on the relations between directed graph parameters. 

\subsection{Dilation Bounds for Spanners of Bounded Directed Path-Width}
We first consider spanners with bounded directed path-width. 

For undirected parameters it is well known that the tree-width of a graph is a lower bound for its path-width. This simply arises from the fact that the definition of tree-width is a generalization of the definition of the path-width. In the directed setting this is not as obvious.

\begin{lemma}
    Let $D$ be a digraph. It holds that $\dtw(D)\leq \dpw(D)$.
\end{lemma}
\begin{proof}
    Let $(X_1,\dots, X_r)$ be a directed path-decomposition of width $k$ for some digraph $D$. We construct a directed tree-decomposition $(T,\mathcal{X},\mathcal{W})$ of width $k$ for $D$ as follows:
    \begin{itemize}
        \item The directed tree will simply be a directed path $T = (\{v_1,\dots,v_r\},\{(v_i,v_{i + 1})\mid 1\leq i<r\})$.
        \item Every $X(v_i)$ will contain the vertices that are newly introduced in the corresponding bag $X_i$. This means $\mathcal{X}(v_i) = X_i-\bigcup_{j = 1}^{i-1}X_j$.
        \item And lastly we set $\mathcal{W}((v_i,v_{i + 1})) = X_i\cap X_{i + 1}$.
    \end{itemize}
    For every $v_i\in V(T)$ we have that $\vert \mathcal{X}(v_i)\cup\bigcup_{e\sim v_i}\mathcal{W}(e)\vert-1\leq \vert X_i\vert-1\leq k$. It is also not hard to see that $\{\mathcal{X}(v_i)\mid v_i\in V(T)\}$ is by construction a partition of $V(D)$. Now suppose that there is some $e=(v_i,v_{i + 1})$ for which there is a closed walk in $D-\mathcal{W}(e)$ containing a vertex of $\mathcal{X}(T_{v_{i + 1}})$ and a vertex of $V(D)-\mathcal{X}(T_{v_{i + 1}})$. This means that there is some edge $(q,p)$ where $p\in V(D)-\mathcal{X}(T_{v_{i + 1}})$ and $q\in\mathcal{X}(T_{v_{i + 1}}) $. For this to be the case $q$ must have been introduced in some $X_{\ell_q}$ for some $\ell_q\geq i + 1$ and $p$ introduced in some $X_{\ell_{p}}$ for some $\ell_p\leq i$. Therefore, $p$ must be contained in every $X_\ell$ for $\ell_p\leq \ell\leq\ell_q$ and hence in $X_i$ and $X_{i + 1}$. Then by definition we have that $p\in \mathcal{W}(e)$, a contradiction.
\end{proof}

Nonetheless, the following corollary holds by Theorem \ref{thm:dtw_lower} and the lemma above.

\begin{corollary} \label{cor-dpw-lower}
Let $d \geq 2$ be a fixed integer. For positive integers $n$ and $k \leq n^{(d-1)/d} \cdot(12d)^{1/d-2}$
there is a set of $n$ points in $\mathbb{R}^d$, so that every directed geometric spanner of directed path-width $k$ on this set has dilation $\Omega(n/k^{d/(d-1)})$.
\end{corollary}

The corresponding upper bound of $\bigO(n/k^{d/(d-1)})$ once again follows from the upper bound for spanners of bounded undirected path-width where every edge of the spanner becomes bidirected and the following lemma.

\begin{lemma}[\cite{YangC08}]
    Let $D$ be a digraph and $G$ its undirected version, then $\dpw(D)=\pw(G)$.\footnote{The proofs shown in \cite{YangC08} use the notation of directed vertex separation number, which is known to be equal to directed path-width.}
\end{lemma}

\begin{corollary} \label{cor-dpw-upper}
Given a set of $n$ points $\sf P\subset\mathbb R^d$ and some positive integer $k\leq n^{1-1/d}$, there is a directed geometric spanner of directed path-width $k$ and dilation $\bigO(n/k^{d/(d-1)})$.
\end{corollary}

\subsection{Dilation Bounds for Spanners of Bounded DAG-Width}
We now show the mentioned bounds for spanners of bounded DAG-Width.

\begin{lemma}[Proposition 36 in \cite{BerwangerDHKO12}]
    Let $D$ be a digraph, then $\dw(D)\leq \dpw(D) + 1$.
\end{lemma}

\begin{corollary} \label{cor-dagw-upper}
Given a set of $n$ points $\sf P\subset\mathbb R^d$ and some positive integer $k\leq n^{1-1/d}$, there is a directed geometric spanner of DAG-width $k$ and dilation $\bigO(n/k^{d/(d-1)})$.
\end{corollary}

A corresponding lower bound can be obtained through the lower bound for spanners of bounded directed tree-width. It was shown in \cite{BerwangerDHKO12} that the directed tree-width of a digraph $D$ is bounded by $3\dw(D) + 1$. The definition of directed tree-width used in \cite{BerwangerDHKO12} was the definition originally given by Johnson et al. in \cite{JRST01}. The definition of directed tree-width used in this paper is a lower bound for the definition of Johnson et al. \cite{kim2025directedtreewidthclosedtaking}.

\begin{corollary} \label{cor-dagw-lower}
Let $d \geq 2$ be a fixed integer. For positive integers $n$ and $k \leq (n^{(d-1)/d} \cdot(12d)^{1/d-2}-1)/3$
there is a set of $n$ points in $\mathbb{R}^d$, so that every directed geometric spanner of DAG-width $k$ on this set has dilation $\Omega(n/k^{d/(d-1)})$.
\end{corollary}

\section{Conclusion and Open Questions}

In this paper we show that for numerous graph parameters, computing a geometric spanner that is bounded in this parameter yields a dilation of $\bigO(n/k^{d/(d-1)})$, which is asymptotically worst case optimal. One parameter that remains open is band-width: This is closely related to path-width, but even more restrictive than cut-width. Thus it is open, whether there exists a spanner of band-width $k$ and dilation $\bigO(n/k^{d/(d-1)})$ for every set of points in $\mathbb{R}^d$.

Further, we consider tree-depth as a graph parameter that behaves differently: While we cannot give a dilation upper bound for spanners with bounded tree-depth, we present an XP-algorithm that, for fixed tree-depth $k$, computes a $2$-approximation spanner with tree-depth $k$ in polynomial time.
It remains open if similar results are possible for other graph parameters. Graphs of tree-width $1$ are trees, and while it is known that computing a tree with minimum dilation is NP-hard \cite{Klein2007, CHEONG2008188}, to our knowledge no approximation algorithms exist. Thus, the open question raised by Eppstein \cite{DBLP:books/el/00/Eppstein00} of finding an approximation algorithm to compute a tree with minimum dilation remains open and consequently, this is open also for spanners of bounded tree-width.
The same is true for paths with minimum dilation \cite{GKM07} and thus for spanners of bounded path-width.

\bibliography{bib}

\newpage

\appendix

\section{Linear Dilation Path-Spanner}

In this section we restate the proof of Matoušek but with respect to geometric spanners using our notation and a slightly more careful analysis of the length and dilation, which leads to improved bounds.

We use $\sigma(\P)$ to denote the length of the longest edge in $\text{EMST}(\P)$. Before we prove the main result we shortly discuss how to compute a path of small length, which we need for the main proof. 
\begin{lemma}
\label{lem:easypath}
    Let $\P\subset \mathbb{R}^d$ be a set of $n$ points and let $p\in \P$ be some point. There is a path starting at $p$ of length $\leq 2n\sigma(\P)$ containing all points of $\P$.
\end{lemma}
\begin{proof}
    Let $T=\text{EMST}(\P)$. $T$ has $n-1$ edges and the length of every edge is at most $\sigma(\P)$, hence the total weight of the EMST is at most $(n-1)\sigma(\P)$. A simple DFS-traversal on $T$ starting at $p$ uses every edge at most two times. Therefore the length of this traversal is at most two times the length of the EMST. By skipping points that are visited multiple times (shortcutting) in this traversal, we obtain a path on $\P$ starting at $p$. By the triangle inequality the total length of this path must be bounded by the length of the DFS-traversal and therefore by $2(n-1)\sigma(\P)$.  
\end{proof}

\matousek*
\begin{proof}
    We will prove the above lemma by proving the following statement by induction on the number of points $n$. 

    \begin{statement}[H1]
            Let $\P\subset \mathbb{R}^d$ be a set of $n$ points. There is a path $P$ of length $2n\sigma(\P)-2\sigma(\P)$ and dilation at most $4n$ that contains all points of $\P$.
    \end{statement}

    To obtain the exact statement of the lemma above one would simply stop the induction at the appropriate step.

    For $n=1$ the statement holds trivially. Let $\vert\P\vert=n>1$, $\sigma=\sigma(\P)$ and let $T=\text{EMST}(\P)$. It is easy to see that $T[\leq \sigma/2]$ is disconnected. Let $\mathcal{K}$ be the set of its connected components. For points $p$, $q$ inside the same connected component we know that the path in $T$ between these points only uses edges of length at most $\sigma/2$.

    We now define an auxiliary rooted tree $S$. The set of vertices of this tree is $\mathcal{K}$. The root $R$ of $S$ is chosen arbitrarily. Components $K,L\in \mathcal{K}$ are connected through an edge if there is an edge between points of the components in $T$. For every component $K\in\mathcal{K}$ we chose a representative point $p_K$. For the root $R$ we chose the point $p_R\in R$ arbitrarily. For the remaining components $K$, $p_K$ is the point of $K$ which has an edge in $T$ to a point in its parent component in $S$. For every $K\in \mathcal{K}$ we have $\sigma(K)\leq \sigma/2$ and $\vert K\vert < n$. By the induction hypothesis there is path $P_K$ of length $\vert K\vert\sigma-\sigma$ and dilation at most $4\vert K\vert$ containing all points of $K$.

    The path for $\P$ will now be built inductively, starting from the leaves of the tree $S$ up to its root, using the paths for the components guaranteed by the induction hypothesis.

    Let $K$ be a node of $S$ and let $S_K$ be the subtree of $S$ induced by $K$. We use $U_K$ to denote the set of all points in the components of the subtree $S_K$.
    We now inductively prove the following statement:

    \begin{statement}[H2]
    There is a path $Q_K$ of length $2\sigma\vert U_K\vert-2\sigma$ and dilation $4\vert U_K\vert$ containing all points of $U_K$.
    \end{statement}
    
    For $K=R$ this proves the inductive step of the top level induction.

    If $K$ is a leaf in $S$ then by (H1) there is a path with the required properties. If $K$ is an inner node then we know that (H2) holds for all of its children in $S$. Let $q_1,\dots q_t$ be the path guaranteed by Lemma \ref{lem:easypath} for $\P=K$ and $v=p_K$. For $j=1,2,\dots,t$, let $M_j=\{K_{j,1}, K_{j,2}, \dots, K_{j,m(j)}\}$ be the set of those components of $S_K$, which are directly connected by an edge in $T$ to the point $q_j$. Let $M=M_1\cup M_2\cup \dots\cup M_q$ i.e $M$ contains all children of $K$ in $S$.

    We now define the path $Q_K$ for $U_K$ by concatenating the paths for the child components as follows:
    \[Q_K = P_K\cdot Q_{K_{1,1}}\cdot Q_{K_{1,2}}\cdots Q_{K_{1,m(1)}}\cdot Q_{K_{2,1}}\cdots Q_{K_{t-1,m(t-1)}}\cdots Q_{K_{t,m(t)}} \]
    The total length of the resulting path can bounded through:
    \[(\sigma \vert K\vert-\sigma)+\sigma \vert K\vert+(2\sigma\vert M\vert-\sigma)+2\sum_{L\in M}(\sigma \vert U_L \vert-\sigma)= 2\sigma\vert K \vert-2\sigma+\sum_{L\in M}2\sigma\vert U_L\vert=2\sigma\vert U_K\vert-2\sigma\]
    Here $\sigma \vert K\vert-\sigma$ is the length of the path $P_K$ and $\sigma \vert K\vert$ is the length of the path $q_1,\dots,q_t$ which in parts bounds the total length of the edges which connect paths for different $M_j$. To be more exact it bounds the portion of the length of these edges that come from the distance between the points $q_1,\dots,q_t$. With $2\sigma \vert M\vert-\sigma$ we account for the distance between points $q_j$ and the representatives $r_{K_{j,1}}, \dots , r_{K_{j,m(j)}}$. Lastly with $2\sum_{L\in M}(\sigma\vert U_L\vert - \sigma)$ we account on one hand for the length of the paths $Q_{1,1},\dots,Q_{t,m(t)}$ and on the other hand the length of the edges inserted between the paths $Q_{1,1},\dots,Q_{t,m(t)}$ minus the length that is attributed to the distance between points $q_1,\dots,q_t$.

    Now let us show the dilation bound. Let $p$ and $q$ be two distinct points of $U_K$. If both $p$ and $q$ lie in the same set $U_L$ for some $L\in M$ or both lie in $K$ then the dilation bound follows from (H2) or (H1) respectively. If that is not the case then $\|pq\|>\sigma/2$. The distance between $p$ and $q$ along the path is at most $2\sigma\vert U_K\vert-2\sigma$, as shown above. Therefore the dilation is at most $4\vert U_K\vert-4$.
\end{proof}

\begin{lemma}
    For a set of $n$ points in $\mathbb{R}^d$ the path in the above construction can be computed in time $\bigO(n^2)$ and in $\bigO(n\log n)$ for $d=2$.
\end{lemma}\begin{proof}
    To construct the above path algorithmically we first compute $T=\text{EMST}(\P)$ this can be done in time $\bigO(n\log n)$ in $\mathbb{R}^2$ otherwise in $\bigO(n^2)$. Then $T$ is traversed starting at an arbitrary point and for each point we store when it was visited in a post-order traversal. We call this number the rank of the point. The stored information is used to combine the components from the leaves to the root as well as the construction of the paths $q_1,\dots ,q_t$ on the root components of the subtrees. This is done in $\bigO(n)$ time. Then the edges of $T$ are sorted into buckets by their length. An edge $e$ is sorted into bucket $B_i$ if $\lceil\log(\|e\|)\rceil=i$.  This $i$ is called the index of the bucket. Let $\mathcal{B}$ be the set of all buckets. For every point $p$ a path is created only containing $p$. We call such a path a component-path. Now we iterate over the buckets, starting with the bucket with the lowest index in increasing order. Every edge in the bucket with index $i$ connects two component-paths. For every component-path we store the smallest rank of a point in the path and use it to identify leaves in the tree induced by the edges of the bucket between the component-paths. Starting from the leaves, the component-paths are merged. For every edge incident to a point of a component-path the children component-paths are connected by edges according to the rank of the endpoints of the edges inside the parent component-path. All of this can be done in time $\bigO(\vert B_i\vert  \log \vert B_i\vert )$. As $\vert\bigcup \mathcal{B} \vert =n$ the total time for merging all component-paths for all buckets is 
    \[\sum_{B_i\in\mathcal{B}}\bigO(\vert B_i\vert\log \vert B_i\vert)=\bigO(n\log n).\qedhere\]
\end{proof}

\section{Definitions and Related Work of Graph Parameters}

\subsection{Undirected Graph Parameters}

\paragraph*{Tree-width and path-width}
Tree-width and path-width as a measure how ``tree-like'' (resp. path-like) have already been defined in the 1980s \cite{DBLP:journals/jct/RobertsonS83}. 
An undirected graph $G=(V_G,E_G)$ has tree-width (path-width) at most $k$, if there is a pair $T, \mathcal{X}$ where $T=(V_T, E_T)$ is a tree (path) and $\mathcal{X}=\{ X_v \mid v \in V_T \}$ is a set of vertex sets $X_v \subseteq V_G$ such that 
\begin{itemize}
	\item $\bigcup_{v\in V_T} X_v = V_G$
	\item For every edge $\{a,b\} \in E_G$ there is a vertex $v \in V_T$, such that $a,b \in X_v$
	\item For every vertex $a \in V_G$, the by vertices $\{v \in V_T \mid a \in X_v\}$ induced subgraph $T$ connected
	\item $|X_v| - 1 \leq k$ for all $v \in V_T$.
\end{itemize}

The pair $(T, \mathcal{X})$ is then called a \emph{tree decomposition} (path decomposition) of $G$ and $k$ is the width of  $(T, \mathcal{X})$.

Many problems which are NP-hard in general are computable in polynomial time on  graphs with bounded tree-width, by using a dynamic programming algorithm on the tree decomposition. Some of these algorithms and a general introduction how to construct them can be found in \cite{cygan2015}. Further, every problem which is describable in monadic second order logic with vertex and edge quantification is solvable in polynomial time for graphs with bounded tree-width \cite{CO00}. 

Path-width, requiring a path structure instead of a tree-like structure, is more restricted than tree-width. 

\paragraph*{Branch-width}

The \emph{branch-width} of a graph has been introduced by Robertson and Seymour in 1991 \cite{RobertsonS91}.

A branch-decomposition of an undirected graph $G=(V_G,E_G)$ is a pair $(T, \tau)$, where $T$ is an
unrooted binary tree and $\tau$ is a bijection from the set of leaves of $T$ to $E_G$. The
order of an edge $e$ of $T$ is the number of vertices $u$ of $G$ such that there are
leaves $t_1$, $t_2$ of $T$ in different components of $T-e$, with $\tau( t_1)$, $\tau(t_2)$ both incident with $u$. The width of $(T, \tau)$ is the maximum order of the edges of $T$, and the branch-width $\bw(G)$ of $G$ is the minimum width of all branch-decompositions of $G$.

\paragraph*{Cut-width}

Cut-width is an even more restricted parameter than path-width~\cite{KORACH199397}. 
It is defined in~\cite{Chung85} as follows.

Let $G=(V_G,E_G)$ be an undirected graph. A numbering $\pi$ of $G$ is defined as a bijection $\pi: V_G\rightarrow \{1,\dots, \vert V_G\vert\}$. The width of a numbering $\pi$ is defined as $\max_{i\in\{1,\dots,\vert V_G\vert\}}\vert\{\{u,v\}\in E_G\mid \pi(u)\leq i < \pi(v)\}\vert$. The cut-width $\cw(G)$ of $G$ is defined as the minimum width among all numberings of $G$.

\paragraph*{Clique-width}
The \emph{clique-width} of a graph has been introduced by Courcelle and Olariu in 1994 \cite{CO00}. The class $\CW_k$ of graphs with clique-width at most $k$ consists of labeled graphs $G=(V,E,\lab)$ with labeling function $\lab : V \rightarrow [k]$ which can be constructed as follows: 

\begin{itemize}
\item Creation of a new vertex  with label $a$, denoted by $\bullet_a$, for some $a \in [k]$ is in $\CW_k$. 
\item Disjoint union of two vertex-disjoint labeled graphs $G$ and $H$
denoted by $G\oplus H$. 
\item For $G=(V,E,\lab) \in \CW_k$, inserting an edge from every vertex with label $a$ to every vertex with label $b$, where $a,b \in [k]$, $a\neq b$, denoted by $\alpha_{a,b} := (V, E', \lab)$ 

 is in $\CW_k$.
\item For $G=(V,E,\lab) \in \CW_k$, change label $a$ into label $b$, denoted by $\rho_{a\to b} = (V,E, \lab')$ with $ \lab'(u) =\lab(u)$ if $ \lab(u) \not= a$ and $\lab'(u) = b$ if $\lab(u)=a$ 
for every $u \in V_G$ is in $\CW_k$.
\end{itemize}

The {\em clique-width} of an unlabeled graph $G=(V,E)$ is the smallest integer
$k$, such that there is a mapping  $\lab : V \to [k]$ such that
the labeled graph  $(V,E,\lab)$ is in $\CW_k$.

An expression $X$ built with the operations defined above is called a {\em clique-width $k$-expression}. The tree structure for the parameter of clique-width is given by the tree structure of that expression.

Also for clique-width, many problems which are NP-hard in general are computable in polynomial or even linear time on graphs with bounded clique-width. However, it can be seen as a more general parameter than tree-width, since there are sparse and dense graphs with bounded clique-width. Courcelle and Olariu \cite{CO00} showed that any problem, which is describable in monadic second order logic with only vertex quantifications is computable in linear time on graphs with bounded clique-width. 

\paragraph*{Rank-width}

The rank-width is closely related to clique-width and was introduced by Oum and Seymour \cite{OumS06}. The definition given here is taken from \cite{Oum17}.

For an undirected graph $G=(V_G,E_G)$ and a subset $X\subseteq V_G$ we define $\rho_G(X)$ to be the rank of a $\vert X\vert \times \vert V(G)-X\vert$ 0-1 matrix $A_X$ over the binary field where the entry of $A_X$ on the $i$-th row and $j$-th column is 1 if and only if the $i$-th vertex in $X$ is adjacent to the $j$-th vertex in $V(G)-X$. If $X=\emptyset$ or $X=V_G$, then $\rho_G(X)=0$.

A rank-decomposition of $G$ is a pair $(T,\tau)$ where $T$ is an unrooted binary tree with at least two nodes and $\tau$ is a bijection from the set of leaves of $T$ to $V_G$. For each edge $e$ of $T$, $T-e$ induces a partition $\{A_e,B_e\}$ of the leaves of $T$ and we define the width of $e$ as $\rho(\tau(A_e))$. The width of a rank-decomposition $(T,\tau)$ is the maximum width of edges in $T$. The rank-width $\rw(G)$ of $G$ is the minimum width among all rank-decompositions of $G$.

\paragraph*{Tree-depth}

While the tree-width measures how ``tree-like'' a graph is and the path-width how ``path-like'' a graph is, the tree-depth could be related to how ``star-like'' a graph is. It has been defined in \cite{Nesetril} recursively as follows: 

The tree-depth of a graph $G$ with connected components $G_1,\dots,G_\ell$ can be defined recursively as:
    \[\td(G)=\begin{cases}
        1& \text{if } |V(G)|=1\\
        1 + \min_{v\in V}\td(G-v) & \text{if $G$ is connected and $|V(G)|>1$}\\
        \max_{i\in\{1,\dots,\ell\}}\td(G_i) & \text{otherwise}
    \end{cases}\]
Based on this definition, it is easy to check that graphs of tree-depth $2$ are forests of stars.

\subsection{Directed Graph Parameters}

\paragraph*{Directed tree-width}
The first directed version for tree-width is the directed tree-width, defined in \cite{JRST01} by Johnson et al. and in several other publications \cite{Ree99,JRST01a,DE14,BG18} equivalently. Although the definition seems very different from the undirected version at first sight, several important properties for undirected tree-width on undirected graphs remain fulfilled using directed tree-width on directed graphs and there has been much research regarding this parameter \cite{DBLP:journals/tcs/Bodlaender98, Adl07, GR19, GR18c, Wie19}.

The notion of directed tree-width we use in our proofs is the one introduced by Giannopoulou, Kawarabayashi, Kreutzer, and Kwon in \cite{DBLP:conf/soda/GiannopoulouKKK22}. Kim, Hatzel and Kreutzer showed in \cite{kim2025directedtreewidthclosedtaking} that this definition is equivalent up to a constant factor to all other common definitions of directed tree-width.

A digraph $D$ has directed tree-width $\dtw(D)$ at most $k$ if there is a triple $\mathcal{T}=(T, \mathcal{X}, \mathcal{W})$, where $T$ is a rooted directed tree, $\mathcal{X}: V(T)\rightarrow 2^{V(D)}$ and $\mathcal{W}: E(T)\rightarrow 2^{V(D)}$ such that: 
\begin{itemize}
     \item $\{\mathcal{X}(t) \mid t\in V(T)\}$ is a partition of $V(D)$ into possibly empty sets such that $\mathcal{X}(r)\neq \emptyset$, where $r$ is the root of $T$
     \item for all $e=(s,t)\in E(T)$, there is no closed walk in $D-\mathcal{W}(e)$ containing a vertex of $\mathcal{X}(T_t)$ and a vertex of $V(D)-\mathcal{X}(T_t)$, and
     \item $\max\{|\mathcal{X}(t)\cup \bigcup_{e\sim t} \mathcal{W}(e)|-1 \mid t\in V(T)\} \leq k$ where $e\sim t$ means that $t$ is an endpoint of the edge $e$.
\end{itemize}
Here, $T_t:=T[\{t'\in V(T)\mid \text{$t'$ is reachable from $t$ by a directed path in $T$}\}]$ for every $t\in V(T)$ and $\mathcal{X}(S):=\bigcup_{t\in V(S)}\mathcal{X}(t)$.
The triple $\mathcal{T}=(T, \mathcal{X}, \mathcal{W})$ is called a \emph{directed tree-decomposition} of $D$ and $k$ is the width of $\mathcal{T}$.

\paragraph*{Directed path-width}
The notion of directed path-width was
introduced by Reed, Seymour, and Thomas around 1995 (cf.\  \cite{Bar06}) and relates to directed
tree-width introduced by Johnson, Robertson, Seymour, and Thomas in
\cite{JRST01}.
It is strongly related to several other graph parameters~\cite{DBLP:journals/mst/GurskiR19}. 

A digraph $D=(V_D, E_D)$ has directed path-width $\dpw(D)$ at most $k$, if there is a sequence $\mathcal{X}=(X_1, \ldots, X_r)$ of subsets of $V$, called {\em bags}, such 
that 

\begin{itemize}
\item $X_1 \cup \ldots \cup X_r ~=~ V_D$, 
\item for each $(u,v) \in E_D$ there is a pair $i \leq j$ such that
  $u \in X_i$ and $v \in X_j$
\item for all $i,j,\ell$ with $1 \leq i < j < \ell \leq r$ it holds
  $X_i \cap X_\ell \subseteq X_j$
  \item $|X_i|-1 \leq k$ for $1 \leq i \leq r$
\end{itemize}
The sequence $\mathcal{X}$ is called directed path decomposition of $D$ and $k$ is the width of $\mathcal{X}$.

\paragraph*{DAG-width}
Though directed tree-width is a natural extension of its undirected version, it does not seem equally algorithmically powerful. Thus, other directed variants of tree-width have been introduced, amongst them the DAG-width \cite{BDHK06,Obd06}, which admits FPT algorithms for several problems such as hamiltonicity and disjoint paths \cite{BerwangerDHKO12}.

A digraph $D=(V_D,E_D)$ has DAG-width at most $k$ if there is a pair $(T, \mathcal{X})$ where $T=(V_T, E_T)$
is a  directed acyclic graph (DAG) and $\mathcal{X} = \{X_u \mid X_u \subseteq V_D, u \in V_T\}$ is a
family of subsets of $V_D$ such that:

\begin{itemize}
	\item $\bigcup_{u \in V_T} X_u = V_D$.
	\item For all vertices $u,v,w \in V_T$ with $u \succcurlyeq_T v \succcurlyeq_T w$,
it holds that $X_u \cap X_w \subseteq X_v$.
	\item For all edges $(u,v) \in E_T$ it holds
that $X_u \cap X_v$ guards $X_{\succcurlyeq_v} \setminus X_u$,
where $X_{\succcurlyeq_v}$ is the union of all sets $X_w$, such that there is a path from $v$ to $w$ in $T$.
For any source $u$, $X_{\succcurlyeq_u}$
is guarded by $\emptyset$. \footnote{Let $V' \subseteq V_D$, then a set $W \subseteq V_D$ {\em guards} $V'$ if for all $(u,v) \in E_D$
it holds that if $u \in V'$ then $v \in V' \cup W$.}
    \item $ |X_u| \leq k$ for all $u\in V_T$
\end{itemize}

\section{Tabular Overview of Dilation Bounds for Different Graph Parameters}

\begin{table}[H]
    \centering
\begin{tabular}{c|c|c|c}

\textbf{Parameter} & \textbf{Dilation upper bound} & \textbf{Dilation lower bound} & \textbf{Reference}\\
\hline
tree-width  & $\bigO(n/k^{d/(d-1)})$  & $\Omega(n/k^{d/(d-1)})$ & \cite{DBLP:conf/compgeom/BuchinRS25} \\
path-width & $\bigO(n/k^{d/(d-1)})$  & $\Omega(n/k^{d/(d-1)})$ & \cref{thm:pwupper}\\
branch-width $\geq 2$ & $\bigO(n/k^{d/(d-1)})$   & $\Omega(n/k^{d/(d-1)})$ &\cref{cutBranchwidth}\\ 
cut-width  & $\bigO(n/k^{d/(d-1)})$   & $\Omega(n/k^{d/(d-1)})$ &\cref{cutBranchwidth} \\
band-width & open  &   $\Omega(n/k^{d/(d-1)})$ & - \\
clique-width $\geq 2$ & 1   & 1 & - \\
rank-width & 1 & 1 & - \\
clique-width $\cap$ planar & $\bigO(n/k^{2})$ (in $\mathbb{R}^2$)   & $\Omega(n/k^{2})$ (in $\mathbb{R}^2$)& \cref{cliqueRankwidth}  \\
rank-width $\cap$ planar & $\bigO(n/k^{2})$ (in $\mathbb{R}^2$)   & $\Omega(n/k^{2})$ (in $\mathbb{R}^2$) & \cref{cliqueRankwidth} \\
tree-depth & - & $1/\varepsilon$ for $\varepsilon<1$ & \cref{thm-treedepth} \\
directed tree-width  & $\bigO(n/k^{d/(d-1)})$  & $\Omega(n/k^{d/(d-1)})$ & \cref{thm:dtw_lower} \\
directed path-width  & $\bigO(n/k^{d/(d-1)})$  & $\Omega(n/k^{d/(d-1)})$ & Cor. \ref{cor-dpw-lower},\ref{cor-dpw-upper} \\
DAG-width  & $\bigO(n/k^{d/(d-1)})$  & $\Omega(n/k^{d/(d-1)})$ & Cor. \ref{cor-dagw-upper},\ref{cor-dagw-lower} \\

\end{tabular}
    \caption{Overview over dilation bounds for (directed) spanners bounded in different parameters where $n$ is the number of points in $\mathbb{R}^d$ and $k$ bounds the width-parameter.}
    \label{tab:placeholder}
\end{table}

\end{document}